\documentclass[12pt]{article}

\usepackage{microtype}
\usepackage{graphicx}
\usepackage{subfigure}
\usepackage{booktabs}
\usepackage{xcolor}
\usepackage{hyperref}
\hypersetup{colorlinks,
            linkcolor=blue,
            citecolor=blue,
            urlcolor=magenta,
            linktocpage,
            plainpages=false}
\usepackage{amsmath}
\usepackage{amssymb}
\usepackage{mathtools}
\usepackage{amsthm}
\usepackage[margin=1in]{geometry}
\usepackage{authblk}
\usepackage{natbib}
\usepackage{thmtools, thm-restate}

\theoremstyle{plain}
\newtheorem{theorem}{Theorem}[section]

\newtheorem{lemma}[theorem]{Lemma}
\newtheorem{corollary}[theorem]{Corollary}
\theoremstyle{definition}
\newtheorem{definition}{Definition}

\theoremstyle{remark}

\def\imghome{./figures}
\newcommand{\mlvec}[1]{\boldsymbol{\mathbf{#1}}}
\newcommand{\partialtheta}[2]{\frac{\partial #1}{\partial \theta_{#2}}}
\newcommand{\diffT}[1]{\frac{d #1}{d t}}

\newcommand{\qnnmeasure}{{\mlvec{M}_0}}
\newcommand{\tildeH}[1]{{\mlvec{H}_{#1}}}
\newcommand{\paramM}{{\mlvec{M}( \mlvec{\theta}( t ) )}}
\newcommand{\varytheta}{{\mlvec{\theta}(t)}}
\newcommand{\generatorH}[1]{{\mlvec{H}^{(#1)}}}
\newcommand{\Kasym}{{\mlvec{K}_\mathsf{asym}}}
\newcommand{\Kpert}{{\mlvec{K}_\mathsf{pert}}}
\newcommand{\varyl}{{L(\varytheta)}}
\newcommand{\varyK}{{\mlvec{K}(\varytheta)}}
\newcommand{\varyM}{{\mlvec{M}(\varytheta)}}
\newcommand{\varyr}{{\mlvec{r}(\varytheta)}}
\newcommand{\compI}{{\mathrm{i}}}
\newcommand{\imagi}{{i}}
\newcommand{\fronorm}[1]{\norm{#1}_{{F}}}
\newcommand{\opnorm}[1]{\norm{#1}_{\mathsf{op}}}
\def\HH{\mlvec{H}}
\def\YY{\mlvec{Y}}

\def\01{\{0,1\}}

\newcommand{\Prob}{{\mathbf{Pr}}}
\newcommand{\tinyspace}{\mspace{1mu}}

\newcommand{\norm}[1]{\left\lVert\tinyspace#1\tinyspace\right\rVert}

\newcommand{\tr}{\operatorname{tr}}
\newcommand{\trnorm}[1]{\norm{#1}_{\tr}}

\def\({\left(}
\def\){\right)}

\def\complex{\mathbb{C}}
\def\real{\mathbb{R}}

\def\<{\langle}
\def\>{\rangle}

\def\S{\mathcal{S}}

\def\1{\mathbbm{1}}
\def\EXP{\mathbb{E}}

\author[1,2]{Xuchen You\thanks{xyou@umd.edu}}
\author[3]{Shouvanik Chakrabarti}
\author[4]{Boyang Chen}
\author[1,2]{Xiaodi Wu\thanks{xwu@cs.umd.edu}}
\affil[1]{\small Department of Computer Science, University of Maryland}
\affil[2]{\small Joint Center for Quantum Information and Computer Science, University of Maryland}
\affil[3]{\small Global Technology Applied Research, JPMorgan Chase \& Co.}
\affil[4]{\small Institute for Interdisciplinary Information Sciences, Tsinghua University}
\title{Analyzing Convergence in Quantum Neural Networks:\\
  Deviations from Neural Tangent Kernels}

\begin{document}
\maketitle

\begin{abstract}
A quantum neural network (QNN) is a parameterized mapping efficiently implementable on near-term Noisy Intermediate-Scale Quantum (NISQ) computers.
It can be used for supervised
learning when combined with classical gradient-based optimizers.
Despite the existing empirical and theoretical investigations, the
convergence of QNN training is not fully understood.
Inspired by the success of the neural tangent kernels (NTKs) in probing into the dynamics of classical neural networks, 
a recent line of works proposes to study
over-parameterized QNNs by examining a quantum version of tangent kernels.
In this work, we study the dynamics of QNNs and show that contrary to popular belief it 
is qualitatively different from
that of any kernel regression: due to the unitarity of quantum operations, there is a non-negligible deviation from the tangent kernel regression derived at the random initialization. 
As a result of the deviation, we prove the at-most sublinear convergence for QNNs with Pauli measurements, which is beyond the explanatory power of any kernel regression dynamics.
We then present the actual dynamics of QNNs in the limit of over-parameterization.
The new dynamics capture the change of convergence rate during training, and implies that the range of measurements is crucial to the fast QNN convergence.


\end{abstract}

\section{Introduction}
\label{sec:intro}
Analogous to the classical logic gates, quantum gates are the basic building
blocks for quantum computing. A variational quantum circuit (also referred to as
an ansatz) is composed of 
parameterized quantum gates. A quantum neural network (QNN) is nothing but an
instantiation of learning with
parametric models using variational quantum circuits and quantum measurements:
A $p$-parameter $d$-dimensional QNN for a dataset $\{\mlvec{x}_{i},y_{i}\}$ is
specified by an encoding $\mlvec{x}_{i} \mapsto \mlvec\rho_{i}$ of the feature vectors into quantum states in an underlying
$d$-dimensional Hilbert space $\mathcal{H}$, a variational circuit
$\mlvec{U}(\mlvec{\theta})$ with real parameters $\mlvec{\theta}\in\real^{p}$, and a quantum
measurement $\mlvec{M}_{0}$. The predicted output $\hat{y}_{i}$ is obtained by
measuring $\mlvec{M}_{0}$ on the output $\mlvec{U}(\mlvec{\theta})\mlvec\rho_{i}\mlvec{U}^{\dagger}(\mlvec\theta)$.
Like deep neural networks, the parameters $\mlvec\theta$ in the variational circuits are
optimized by gradient-based methods to minimize an objective function that
measures the misalignments of the predicted outputs and the ground truth labels.

With the recent development of quantum
technology, the near-term Noisy Intermediate-Scale Quantum (NISQ) \citep{preskill2018quantum}
computer has become an important platform for demonstrating quantum advantage with practical applications.
As a hybrid of classical optimizers and quantum representations, QNNs is a
promising candidate for demonstrating such advantage on quantum computers
available to us in the near future: quantum machine learning models are proved
to have a margin over the classical counterparts in terms of the expressive
power due the to the exponentially large Hilbert space of
quantum states \citep{huang2021power, anschuetz2022critical}. On the other hand
by delegating the optimization procedures to classical computers, the hybrid
method requires significantly less quantum resources, which is crucial for
readily available quantum computers with limited coherence time and error
correction. There have been proposals of QNNs~\citep{Dunjko_2018,schuld2019quantum} for
classification~\citep{farhi2018classification,Romero_2017} and generative
learning~\citep{Lloyd_2018,Zoufal_2019,chakrabarti2019wass}.

Despite their potential there are challenges in the practical deployment of
QNNs. Most notably, the optimization problem for training QNNs can be highly
non-convex. The landscape of QNN training may be swarmed with spurious local
minima and saddle points that can trap gradient-based optimization methods
\citep{you2021exponentially, anschuetz2022quantum}.
QNNs with large dimensions also suffer from a phenomenon called the \emph{barren
plateau}~\citep{mcclean2018barren}, where the gradients of the parameters vanish
at random intializations, making convergence slow even in a trap-free landscape.
These difficulties in training QNNs, together with the challenge of
classically simulating QNNs at a decent scale, calls for a theoretical
understanding of the convergence of QNNs.

\paragraph{Neural Tangent Kernels}
Many of the theoretical difficulties in understanding QNNs have also been
encountered in the study of classical deep neural networks: despite the landscape of neural networks being non-convex and susceptible to
spurious local minima and saddle points, it has been empirically observed that the training errors decays exponentially in the training time~\citep{livni2014on, arora2019exact} in the highly \emph{over-parameterized} regime with sufficiently many number of trainable parameters.
This phenomenon is theoretically explained by connecting the training dynamics of neural networks to the kernel regression:
the kernel regression model generalizes the
linear regression by equipping the linear model with non-linear feature
maps. Given a training set
$\{\mlvec{x}_{j}, y_{j}\}_{j=1}^{m}\subset \mathcal{X}\times\mathcal{Y}$ and a
non-linear feature map $\phi:\mathcal{X}\rightarrow \mathcal{X}'$ mapping the
features to a potentially high-dimensional feature space $\mathcal{X}'$. The kernel
regression solves for the optimal weight $\mlvec{w}$ that minimizes the
mean-square loss
$\frac{1}{2m}\sum_{j=1}^{m}(\mlvec{w}^{T}\phi(\mlvec{x}_{j})-y_{j})^{2}$. The
name of kernel regression stems from the fact that the optimal hypothesis
$\mlvec{w}$ depends on the high-dimensional feature vectors
$\{\phi(\mlvec{x}_{j})\}_{j=1}^{m}$ through a $m\times m$ \emph{kernel} matrix
$\mlvec{K}$, such that $K_{ij}=\phi(\mlvec{x}_{i})^{T}\phi(\mlvec{x}_{j})$. The kernel regression enjoys a linear convergence (i.e. the mean square loss decaying exponentially over time) when $\mlvec{K}$ is positive definite.

The kernel matrix associated with a neural network is determined by tracking how the predictions for each training sample evolve jointly at random initialization. The study of the neural network convergence then reduces to characterizing the corresponding kernel matrices (the neural tangent kernel, or the NTK). 
In addition to the convergence results, NTK also serves as a tool for studying other aspect of neural networks including generalization \citep{canatar2021spectral, chen2020generalized} and stability \citep{bietti2019inductive}.

The key observation that justifies the study of neural networks with neural tangent kernels, is that the NTK becomes a constant (over time) during training in the limit of infinite layer widths. 
This has been theoretically established starting with the analysis of wide fully-connected neural networks~\citep{jacot2018neural,arora2019exact,chizat2018lazy} and later generalized to a variety of architectures (e.g. \citet{allenzhu2019convergence}).

\paragraph{Quantum NTKs}
Inspired by the success of NTKs, recent years have witnessed multiple works attempting to associate over-parameterized QNNs to kernel regression. 
Along the line there are two types of studies.
The first category investigates and compares the properties of the ``quantum'' kernel induced by the quantum encoding of classical features, where $K_{ij}$ associated with the $i$-th and $j$-th feature vectors $\mlvec{x}_{i}$ and $\mlvec{x}_{j}$ equals $\tr(\mlvec{\rho_{i}\mlvec{\rho}_{j}})$ with $\mlvec{\rho}_{i}$ and $\mlvec{\rho}_{j}$ being the quantum state encodings, without referring to the dynamics of training~\citep{schuld2019quantum, huang2021power, liu2021representation}.
The second category seeks to directly establish the quantum version of NTK for QNNs by examining the evolution of the model predictions at random initialization, which is the recipe for calculating the classical NTK in \citet{arora2019exact}:
\citet{shirai2021quantum} empirically evaluates the direct training of the quantum NTK instead of the original QNN formulation.
On the other hand, by analyzing the time derivative of the quantum NTK at initialization, \citet{liu2022analytic} conjectures that in the limit of over-parameterization, the quantum NTK is a constant over time and therefore the dynamics reduces to a kernel regression.

Despite recent efforts, a rigorous answer remains evasive whether the quantum NTK is a constant during training for over-parameterized QNNs. We show that the answer to this question is indeed, surprisingly negative: as a result of the unitarity of quantum circuits, there is a finite change in the conjectured quantum NTK as the training error decreases, even in the the limit of over-parameterization.

\paragraph{Contributions}
In this work, we focus on QNNs equipped with the mean square loss, trained using gradient flow, following~\citet{arora2019exact}.
In Section~\ref{sec:qnn-dynamics}, we show that, despite the formal resemblance to kernel
regression dynamics, the over-parameterized QNN does
not follow the dynamics of \emph{any} kernel regression due to the unitarity: for the widely-considered setting of classifications with Pauli
measurements, we show that the objective function at time $t$ decays at most as
a polynomial function of $1/t$ (Theorem~\ref{thm:sublinear-convergence}). This
contradicts the dynamics of any kernel regression with a positive definite kernel, which exhibits
convergence with $L(t)\leq L(0)\exp(-ct)$ for some  positive constant $c$. We also identify the true asymptotic dynamics of QNN training as regression with a time-varying Gram matrix $\Kasym$ (Lemma~\ref{lm:resid_dyn_decomp}), and show rigorously that the real dynamics concentrates to the asymptotic one in the limit $p \rightarrow \infty$ (Theorem~\ref{thm:qnn-mse-linear}). This reduces the problem of investigating QNN convergence to studying the convergence of the asymptotic dynamics governed by $\Kasym$. 

We also consider a model of QNNs where the final measurement is post-processed by a linear scaling.
In this setting, we provide a complete analysis of the convergence of the asymptotic dynamics in the case of $1$ training sample (Corollary~\ref{cor:onesample}), and provide further theoretical evidence of convergence in the neighborhood of most global minima when the number of samples $m > 1$ (Theorem~\ref{thm:smallest_eig}). These theoretical evidences are supplemented with an empirical study that demonstrates in generality, the convergence of the asymptotic dynamics when $m \ge 1$. Coupled with our proof of convergence, these form the strongest concrete evidences of the convergence of training for over-parameterized QNNs.

\paragraph{Connections to previous works}
Our result extends the existing literature on QNN landscapes (e.g. \citet{anschuetz2022critical,russell2016quantum}) and looks into the training dynamics, which allows us to characterize the rate of convergence and to show how the range of the measurements affects the convergence to global minima.
The dynamics for over-parameterized QNNs proposed by us can be reconciled with the existing calculations of quantum NTK as follows: in the regime of over-parameterization, the QNN dynamics coincides with the quantum NTK dynamics conjectured in \citet{liu2022analytic} at random initialization; yet it deviates from quantum NTK dynamics during training, and the deviation does not vanish in the limit of $p\rightarrow\infty$.


\section{Preliminaries}
\label{sec:prelim}
\paragraph{Empirical risk minimization (ERM)}
A supervised learning problem is specified by a joint distribution $\mathcal{D}$
over the feature space $\mathcal{X}$ and the label space $\mathcal{Y}$, and a
family $\mathcal{F}$ of mappings from $\mathcal{X}$ to $\mathcal{Y}$ (i.e. the
hypothesis set).
The goal is to find an $f\in\mathcal{F}$ that well predicts the label $y$ given the feature $\mlvec{x}$ in
expectation, for pairs of $(\mlvec{x}, y)\in\mathcal{X}\times \mathcal{Y}$ drawn $i.i.d.$
from the distribution $\mathcal{D}$.

Given a training set $\mathcal{S}=\{\mlvec{x}_{j}, y_{j}\}_{j=1}^{m}$ composed of $m$ pairs of
features and labels, we search for the optimal $f\in\mathcal{F}$ by the
\emph{empirical risk minimization} (ERM): let $\ell$ be a loss function
$\ell: \mathcal{Y} \times \mathcal{Y} \rightarrow \real$, ERM finds an $f\in \mathcal{F}$ that minimizes the average loss:
  $\min_{f\in\mathcal{F}}\frac{1}{m}\sum_{i=1}^{m}\ell(\hat{y}_{i}, y_{i}), \text{
  where }\hat{y}_{i} = f(\mlvec{x}_{i})$.
We focus on the common choice of the \emph{square loss} $\ell(\hat{y}, y) = \frac{1}{2}(\hat{y} - y)^{2}$.

\paragraph{Classical neural networks} A popular choice of the hypothesis
set $\mathcal{F}$ in modern-day machine learning is the \emph{classical neural networks}. A
vanilla version of the $L$-layer feed-forward neural network takes the form
$f(x; \mlvec{W}_{1}, \cdots, \mlvec{W}_{L})
  = \mlvec{W}_{L}\sigma(\cdots \mlvec{W}_{2}\sigma(\mlvec{W}_{1}\sigma(x))\cdots )$,
where $\sigma(\cdot)$ is a non-linear activation function, and for all $l\in[L]$, $\mlvec{W}_{l} \in \real^{d_{l}\times d_{l-1}}$ is the
weights in the $l$-th layer, with $d_{L} = 1$ and $d_{0}$ the same as the dimension
of the feature space $\mathcal{X}$. It has been shown that, in the limit
$\min_{l=1}^{L-1}d_{l} \rightarrow \infty$, the training of neural networks with
square loss is close to kernel learning, and therefore enjoys a linear
convergence rate
\citep{jacot2018neural, arora2019exact,
  allenzhu2019convergence, oymak2020toward}.
\paragraph{Quantum neural networks} Quantum neural networks is a family of
parameterized hypothesis set analogous to its classical counterpart. At a high level, it has the layered-structure like a
classical neural network. At each layer, a linear transformation acts on the
output from the last layer. A quantum neural network
is different from its classical counterpart in the following three aspects.

\paragraph{(1) Quantum states as inputs} A $d$-dimensional quantum
state is represented by a \emph{density matrix} $\mlvec{\rho}$, which is a
positive semidefinite $d\times d$ Hermitian with trace $1$. A state is said to
be pure if $\mlvec\rho$ is rank-$1$. Pure states can therefore be equivalently
represented by a state vector $\mlvec{v}$ such that $\mlvec\rho = \mlvec{v}\mlvec{v}^{\dagger}$.
The inputs to QNNs are quantum states. They can either be drawn as samples from a quantum-physical problem or be the encodings of classical feature vectors.

\paragraph{(2) Parameterization}
In classical neural networks, each layer is
composed of a linear transformation and a non-linear activation, and the matrix
associated with the linear transformation can be directly optimized at each entry. In
QNNs, the entries of each linear transformation can not be directly manipulated. Instead
we update parameters in a variational ansatz to update the
linear transformations. More concretely, a general $p$-parameter ansatz $\mlvec{U}(\mlvec\theta)$ in a
$d$-dimensional Hilbert space can be specified by a set of $d\times d$ unitaries
$\{\mlvec{U}_{0}, \mlvec{U}_{1}, \cdots, \mlvec{U}_{p}\}$ and a set of
non-zero $d\times d$ Hermitians $\{\generatorH{1}, \generatorH{2}, \cdots, \generatorH{p}\}$ as
\begin{align}
  & \mlvec{U}_p\exp(-i\theta_p\generatorH{p})\mlvec{U}_{p-1}\exp(-i\theta_{p-1}\generatorH{p-1})
  \cdots\exp(-i\theta_{2}\generatorH{2}) \mlvec{U}_1\exp(-i\theta_1\generatorH{1})\mlvec{U}_0. \label{eq:general-ansatz}
\end{align}
Without loss of generality, we  assume that $\tr(\generatorH{l}) = 0$. This
is because adding a Hermitian proportional to $\mlvec{I}$ on the generator
$\generatorH{l}$ does not change the density matrix of the output states.
Notice that most $p$-parameter ansatze
$\mlvec{U}:\real^{p}\rightarrow \complex^{d\times d}$ can be expressed as
Equation~\ref{eq:general-ansatz}. One exception may be the anastz design with
intermediate measurements (e.g. \citet{cong2019quantum}).
In Section~\ref{sec:theory-fast}, we will also consider the periodic
anastz:
\begin{definition}[Periodic ansatz]
  \label{def:partial-ansatz}
  A $d$-dimensional $p$-parameter periodic anasatz  $\mlvec{U}(\mlvec\theta)$ is defined as
  \begin{align}
    \mlvec{U}_p\exp(-i\theta_p\mlvec{H})\cdot \cdots \cdot \mlvec{U}_1\exp(-i\theta_1\mlvec{H})\mlvec{U}_0, \label{eq:partial-ansatz}
  \end{align}
  where $\mlvec{U}_l$ are sampled $i.i.d.$ with respect to the Haar measure over
  the special unitary group $SU(d)$, and $\mlvec{H}$ is a non-zero trace-$0$ Hermitian.
\end{definition}
Up to a unitary transformation, the periodic ansatz is equivalent to an ansatz
in Line~(\ref{eq:general-ansatz}) where $\{\generatorH{l}\}_{l=1}^{p}$
sampled as $\mlvec{V}_{l}\mlvec{H}\mlvec{V}_{l}^{\dagger}$ with $\mlvec{V}_{l}$
being haar random $d\times d$ unitary matrices.
Similar ansatze have been considered in
\citet{mcclean2018barren,anschuetz2022critical, you2021exponentially, ourvqe2022}.

\paragraph{(3) Readout with measurements} Contrary to classical neural networks, the readout from a QNN
requires performing quantum \emph{measurements}. A measurement is specified by a
Hermitian $\mlvec{M}$.
The outcome of measuring a quantum state $\mlvec\rho$ with a measurement
$\mlvec{M}$ is $\tr(\mlvec\rho \mlvec{M})$, which is a linear function of $\mlvec\rho$. A common choice is the Pauli measurement: Pauli matrices are $2\times 2$ Hermitians that are also unitary. The Pauli measurements are tensor products of Pauli matrices, featuring eigenvalues of $\pm 1$.

A common choice is the Pauli measurement: Pauli matrices are $2\times 2$ Hermitians that are also unitary:
    \begin{align*}
        \sigma_X=\begin{bmatrix}0 & 1\\ 1 & 0\end{bmatrix},
        \sigma_Y=\begin{bmatrix}0 & -\imagi\\ \imagi & 0\end{bmatrix},
        \sigma_Z=\begin{bmatrix}1 & 0\\ 0 & -1\end{bmatrix}.
    \end{align*}
The Pauli measurements are tensor products of Pauli matrices, featuring eigenvalues of $\pm 1$.
\paragraph{ERM of quantum neural network.}
We focus on quantum neural networks equipped with the mean-square loss. Solving the
ERM for a dataset
$\S:=\{(\mlvec\rho_{j}, y_{j})\}_{j=1}^{m}\subseteq (\complex^{d\times d}\times \real)^{m}$
involves optimizing the objective function
$\min_{\mlvec\theta}L(\mlvec\theta):=\frac{1}{2m}\sum_{j=1}^{m}\big(\hat{y}_{j}(\mlvec\theta) - y_{j}\big)^{2}$,
where
$\hat{y}_{j}(\mlvec\theta) =\tr(\mlvec\rho_{j}\mlvec{U}^{\dagger}(\mlvec\theta)\qnnmeasure\mlvec{U}(\mlvec\theta))$
for all $j\in[m]$ with $\qnnmeasure$ being the quantum measurement and $\mlvec{U}(\mlvec\theta)$ being the variational ansatz. Typically, a QNN is trained by optimizing the ERM
objective function by gradient descent: at the $t$-th iteration, the parameters
are updated as
$\mlvec{\theta}(t+1) \leftarrow \mlvec{\theta}(t) - \eta \nabla L(\mlvec{\theta}(t))$,
where $\eta$ is the learning rate; for sufficiently small $\eta$, the
dynamics of gradient descent reduces to that of the gradient flow:
$d\mlvec{\theta}(t)/dt=-\eta \nabla L(\mlvec{\theta}(t))$. Here we focus on the
gradient flow setting following \citet{arora2019exact}.

\paragraph{Rate of convergence} In the optimization literature, the rate of
convergence describes how fast an iterative algorithm approaches an (approximate) solution. For
a general function $L$ with variables $\mlvec\theta$, let $\mlvec\theta(t)$ be
the solution maintained at the time step $t$ and  $\mlvec\theta^{\star}$ be the
optimal solution. The algorithm is said to be converging
\emph{exponentially fast} or at a \emph{linear rate} if
  $L(\mlvec\theta(t)) - L(\mlvec\theta^{\star}) \leq
  \alpha\exp(-c t)$
for some constants $c$ and $\alpha$.
In contrast, algorithms with the sub-optimal gap
$L(\mlvec\theta(t)) - L(\mlvec\theta^{\star})$ decreasing slower than
exponential are said to be converging with a \emph{sublinear} rate
(e.g.  $L(\mlvec\theta(t)) - L(\mlvec\theta^{\star})$ decaying with $t$ as a
polynomial of $1/t$).
We will mainly consider the setting where
$L(\mlvec\theta^{\star}) = 0$ (i.e. the \emph{realizable} case) with continuous time $t$.

\paragraph{Other notations}
We use $\opnorm{\cdot}$,
$\fronorm{\cdot}$ and $\trnorm{\cdot}$ to denote the operator norm (i.e. the
largest eigenvalue in terms of the absolute values), Frobenius norm and the
trace norm of matrices; we use $\norm{\cdot}_{p}$ to denote the $p$-norm of
vectors, with the subscript omitted for $p=2$. We use $\tr(\cdot)$ to denote the trace operation.


\section{Deviations of QNN Dynamics from NTK}
\label{sec:qnn-dynamics}
Consider a regression model on an $m$-sample training set:
for all $j\in[m]$, let $y_{j}$ and $\hat{y}_j$ be the label and the model prediction of the $j$-th sample. The \emph{residual} vector $\mlvec{r}$ is a
$m$-dimensional vector with $r_{j}:= y_{j}-\hat{y}_{j}$.
The dynamics of the kernel regression is signatured by the first-order linear
dynamics of the residual vectors: let $\mlvec{w}$ be the learned model
parameter, and let $\phi(\cdot)$ be the fixed non-linear map. Recall
that the kernel regression minimizes
$L(\mlvec{w})=\frac{1}{2m}\sum_{j=1}^{m}(\mlvec{w}^{T}\phi(\mlvec{x}_{j})-y_{j})^{2}$
for a training set $\mathcal{S}=\{(\mlvec{x}_{j}, y_{j})\}_{j=1}^{m}$, and the
gradient with respect to $\mlvec{w}$ is
$\frac{1}{m}\sum_{j=1}^{m}(\mlvec{w}^{T}\phi(\mlvec{x}_{j})-y_{j})\phi(\mlvec{x}_{j}) = -\frac{1}{m}\sum_{j=1}^{m}r_{j}\phi(\mlvec{x}_{j})$.
Under the gradient flow with learning rate $\eta$, the weight $\mlvec{w}$ updates as
$\frac{d\mlvec{w}}{dt} = \frac{\eta}{m}\sum_{j=1}^{m}r_{j}\phi(\mlvec{x}_{j})$,
and the $i$-th entry of the residual vector updates as
$dr_{i}/dt = -\phi(\mlvec{x}_{i})^{T}\frac{d\mlvec{w}}{dt} = -\frac{\eta}{m}\sum_{j=1}^{m}\phi(\mlvec{x}_{i})^{T}\phi(\mlvec{x}_{j})r_{j}$,
or more succinctly $d\mlvec{r}/dt = -\frac{\eta}{m}\mlvec{K}\mlvec{r}$ with
$\mlvec{K}$ being the kernel/Gram matrix defined as $K_{ij}=\phi(\mlvec{x}_{i})^{T}\phi(\mlvec{x}_{j})$ (see also
\citet{arora2019exact}). Notice that the kernel matrix $\mlvec{K}$ is a constant
of time and is independent of the weight $\mlvec{w}$ or the labels.

\paragraph{Dynamics of residual vectors}
We start by characterizing the dynamics of the residual vectors
for the general form of $p$-parameter QNNs and highlight the limitation of
viewing the over-parameterized QNNs as kernel regressions.
Similar to the kernel regression,
$\frac{dr_{j}}{dt} = -\frac{d \hat{y}_{j}}{dt}= -\tr(\mlvec{\rho}_{j}\frac{d}{dt}\mlvec{U}^{\dagger}(\mlvec\theta(t))\qnnmeasure\mlvec{U}(\mlvec\theta(t)))$
in QNNs. We derive the following dynamics of $\mlvec{r}$ by tracking the
parameterized measurement
$\mlvec{M}(\mlvec{\theta})=\mlvec{U}^{\dagger}(\mlvec\theta)\mlvec{M}_{0}\mlvec{U}(\mlvec\theta)$
as a function of time $t$.
\begin{restatable}[Dynamics of the residual vector]{lemma}{qnndynamics}
\label{lm:qnndynamics}
Consider a QNN instance with an ansatz $\mlvec{U}(\mlvec\theta)$ defined as in
Line~(\ref{eq:general-ansatz}), a training
dataset $\mathcal{S} = \{(\mlvec\rho_{j},y_{j})\}_{j=1}^{m}$, and a measurement
$\qnnmeasure$. Under the gradient flow for the objective function
  $L(\mlvec\theta)=\frac{1}{2m}\sum_{j=1}^{m}\big(\tr(\mlvec{\rho}_{j}\mlvec{U}^{\dagger}(\mlvec\theta)\qnnmeasure\mlvec{U}(\mlvec\theta))-y_{j}\big)^{2}$
  with learning rate $\eta$,
the residual vector $\mlvec{r}$ satisfies the differential equation
\begin{align}
  \label{eq:qnn-residue-kernel}
  \diffT{\mlvec{r}(\varytheta)} = -\frac{\eta}{m}\mlvec{K}(\paramM) \mlvec{r}(\varytheta),
\end{align}
where $\mlvec{K}$ is a positive semi-definite matrix-valued function of the
parameterized measurement. The $(i,j)$-th element of $\mlvec{K}$ is defined as
\begin{align}
  \sum_{l=1}^p\big(
  \tr\big(\compI[\paramM,\mlvec{\rho}_{i}]\tildeH{l}\big)
  \tr\big(\compI[\paramM,\mlvec{\rho}_{j}]\tildeH{l}\big)
  \big).
\end{align}
Here $\tildeH{l}:=\mlvec{U}_{0}^{\dagger}\mlvec{U}_{1:l-1}^{\dagger}(\mlvec\theta)\generatorH{l}\mlvec{U}_{1:l-1}(\mlvec\theta)\mlvec{U}_{0}$,
is a function of $\mlvec\theta$
with $ \mlvec{U}_{1:r}(\mlvec\theta)$ being the shorthand for
$\mlvec{U}_{r}\exp(-i\theta_{r}\generatorH{r})\cdots \mlvec{U}_{1}\exp(-i\theta_{1}\generatorH{1})$.
\end{restatable}

While Equation~(\ref{eq:qnn-residue-kernel}) takes a similar form to that of the kernel
regression, the matrix $\mlvec{K}$ is \emph{dependent} on the parameterized measurement
$\mlvec{M}(\mlvec\theta)$. This is a consequence of the unitarity: consider an alternative
parameterization, where the objective function
$\mlvec{L}(\mlvec{M})=\frac{1}{2m}\sum_{j=1}^{m}\big(\tr(\mlvec\rho_{j}\mlvec{M})-y_{j}\big)^{2}$
is optimized over all Hermitian matrices $\mlvec{M}$. It can be easily verified that the
corresponding dynamics is exactly the kernel regression with $K_{ij}=\tr(\mlvec{\rho}_{i}\mlvec{\rho}_{j})$.

Due to the unitarity of the evolution of quantum states, the spectrum of eigenvalues
of the parameterized measurement $\mlvec{M}(\mlvec\theta)$ is required to remain the
same throughout training.
In the proof of Lemma~\ref{lm:qnndynamics} (deferred to
Section~\ref{subsec:qnndynamics_proof} in the appendix),
we see that the derivative of $\mlvec{M}(\mlvec\theta)$ takes the form of a linear
combination of commutators $i[\mlvec{A}, \mlvec{M}(\mlvec\theta)]$ for some Hermitian $\mlvec{A}$.
As a result, the traces of the $k$-th matrix powers
$\tr(\mlvec{M}^{k}(\mlvec\theta))$ are constants of time for any integer $k$, since
$d\tr(\mlvec{M}^{k}(\mlvec\theta))/dt = k\tr(\mlvec{M}^{k-1}(\mlvec\theta)d\mlvec{M}(\mlvec{\theta})/dt) = k\tr(\mlvec{M}^{k-1}(\mlvec\theta)i[\mlvec{A}, \mlvec{M}(\mlvec{\theta})])=0$
for any Hermitian $\mlvec{A}$.
The spectrum of eigenvalues remains unchanged because the
coefficients of the characteristic polynomials of $\mlvec{M}(\mlvec\theta)$ is
completely determined by the traces of matrix powers. On the contrary, the eigenvalues are in
general not preserved for $\mlvec{M}$ evolving under the kernel regression.

Another consequence of the unitarity constraint is that a QNN can not
make predictions outside the range of the eigenvalues of $\qnnmeasure$, while
for the kernel regression with a strictly positive definite kernel, the model can
(over-)fit training sets with arbitrary label assignments. Here we further show that
the unitarity is pronounced in a typical QNN instance where the predictions are
within the range of the measurement.


\paragraph{Sublinear convergence in QNNs}
One of the most common choices for designing QNNs is to use a (tensor product of)
Pauli matrices as the measurement (see e.g.
\citet{farhi2018classification,Dunjko_2018}). Such a choice
features a measurement $\qnnmeasure$ with eigenvalues $\{\pm 1\}$ and trace zero. Here we
show that in the setting of supervised learning on pure states with Pauli measurements,
the (neural tangent) kernel regression is insufficient to capture the convergence of
QNN training. For the kernel regression with a positive definite kernel $\mlvec{K}$,
the objective function $L$ can be expressed as
$\frac{1}{2m}\sum_{j=1}^{m}(\hat{y}_{j}-y_{j})^{2} = \frac{1}{2m}\mlvec{r}^{T}\mlvec{r}$; under the kernel dynamics of
$\frac{d\mlvec{r}}{dt} = - \frac{\eta}{m} \mlvec{K}\mlvec{r}$, it is easy to verify that
$\frac{d \ln L}{dt} = -\frac{2\eta}{m}\frac{\mlvec{r}^{T}\mlvec{K}\mlvec{r}}{\mlvec{r}^{T}\mlvec{r}}\leq -\frac{2\eta}{m}\lambda_{\min}(\mlvec{K})$
with $\lambda_{\min}(\mlvec{K})$ being the smallest eigenvalue of $\mlvec{K}$.
This
indicates that $L$ decays at a linear rate, i.e. $L(T) \leq L(0) \exp(-\frac{2\eta}{m}\lambda_{\min}(\mlvec{K})T)$.
In contrast,
we show that the rate of convergence of the QNN dynamics \emph{must} be sublinear,
slower than the linear convergence rate predicted by the kernel regression model
with a positive definite kernel.
\begin{restatable}[No faster than sublinear convergence]{theorem}{qnnsublinear}
 \label{thm:sublinear-convergence}
 Consider a QNN instance with a training set
 $\mathcal{S}=\{(\mlvec{\rho}_{j}, y_{j})\}$ such that $\mlvec{\rho}_{j}$ are
 pure states and $y_{j}\in\{\pm 1\}$, and a
 measurement $\qnnmeasure$ with eigenvalues in $\{\pm 1\}$. Under the gradient
 flow for the objective function $L(\mlvec\theta)=\frac{1}{2m}\sum_{j=1}^{m}\tr(\mlvec\rho_{j}\mlvec{M}(\mlvec\theta)-y_{j})^{2}$, for any ansatz $\mlvec{U}(\mlvec\theta)$ defined
in Line~(\ref{eq:general-ansatz}), $L$ converges to $zero$ at most at a
sublinear convergence rate.
More concretely, for $\mlvec{U}(\mlvec\theta)$ generated by $\{\generatorH{l}\}_{l=1}^{p}$,
let $\eta$ be the learning rate and $m$ be the sample size, the objective
function at time $t$:
\begin{align}
L(\mlvec{\theta}(t)) \geq 1 / (c_{0} + c_{1}t)^{2}.
\end{align}
Here the constant $c_{0} = 1/\sqrt{L(\mlvec\theta(0))}$ depends on the objective function at
initialization, and $c_{1}=12\eta\sum_{l=1}^{p}\opnorm{\generatorH{l}}^{2}$.
\end{restatable}
The constant $c_{1}$ in the theorem depends on the number of parameters $p$
through $\sum_{l=1}^{p}\opnorm{\generatorH{l}}^{2}$ if the operator norm of
$\generatorH{l}$ is a constant of $p$. We can get rid of the dependency on $p$
by scaling the learning rate $\eta$ or changing the time scale, which does not
affect the sublinearity of convergence.

By expressing the objective function $L(\varytheta)$ as
$\frac{1}{2m}\mlvec{r}(\varytheta)^{T}\mlvec{r}(\varytheta)$,
Lemma~\ref{lm:qnndynamics} indicates that the decay of
$\frac{d L(\varytheta)}{dt}$ is lower-bounded by
$\frac{-2\eta}{m}\lambda_{\max}(\mlvec{K}(\varytheta)) L(\varytheta)$, where
$\lambda_{\max}(\cdot)$ is the largest eigenvalue of a Hermitian matrix.
The full proof of Theorem~\ref{thm:sublinear-convergence} is deferred to
Section~\ref{subsec:slow_proof}, and follows from the fact that
when the QNN prediction for an input state $\mlvec{\rho}_{j}$ is close to the ground truth
$y_{j}= 1$ or $-1$, the diagonal entry $K_{jj}(\varytheta)$ vanishes. As a
result the largest eigenvalue $\lambda_{\max}(\mlvec{K}(\varytheta))$ also
vanishes as the objective function $L(\varytheta)$ approaches $0$ (which is the
global minima). Notice the sublinearity of convergence is independent of the
system dimension $d$, the choices of $\{\generatorH{l}\}_{l=1}^{p}$ in
$\mlvec{U}(\mlvec\theta)$ or the number of parameters $p$. This means that
the dynamics of QNN training is completely different from kernel regression even
in the limit where $d$ and/or $p\rightarrow\infty$.

\paragraph{Experiments: sublinear QNN convergence} To support
Theorem~\ref{thm:sublinear-convergence}, we simulate the training of QNNs using
$\qnnmeasure$ with eigenvalues $\pm 1$. For dimension $d=32$ and $64$, we
randomly sample four $d$-dimensional pure states that are orthogonal, with two
of samples labeled $+1$ and the other two labeled $-1$.
The training curves (plotted under the log scale) in Figure~\ref{fig:scale1_varyp} flattens as $L$
approaches $0$, suggesting the rate of convergence $-d\ln L/dt$ vanishes around
global minima, which is a signature of the sublinear convergence. Note that the
sublinearity of convergence is independent of the number of parameters $p$.
For gradient flow or gradient descent with sufficiently small step-size, the
scaling of a constant learning rate $\eta$ leads to a scaling of time $t$ and
does not fundamentally change the (sub)linearity of the convergence. For the
purpose of visual comparison, we scale $\eta$ with $p$ by choosing the learning
rate as $10^{-3} / p$. For more details on the experiments, please refer to
Section~\ref{sec:app_exp}.

\begin{figure}[!htbp]
  \centering
  \includegraphics[width=.7\linewidth]{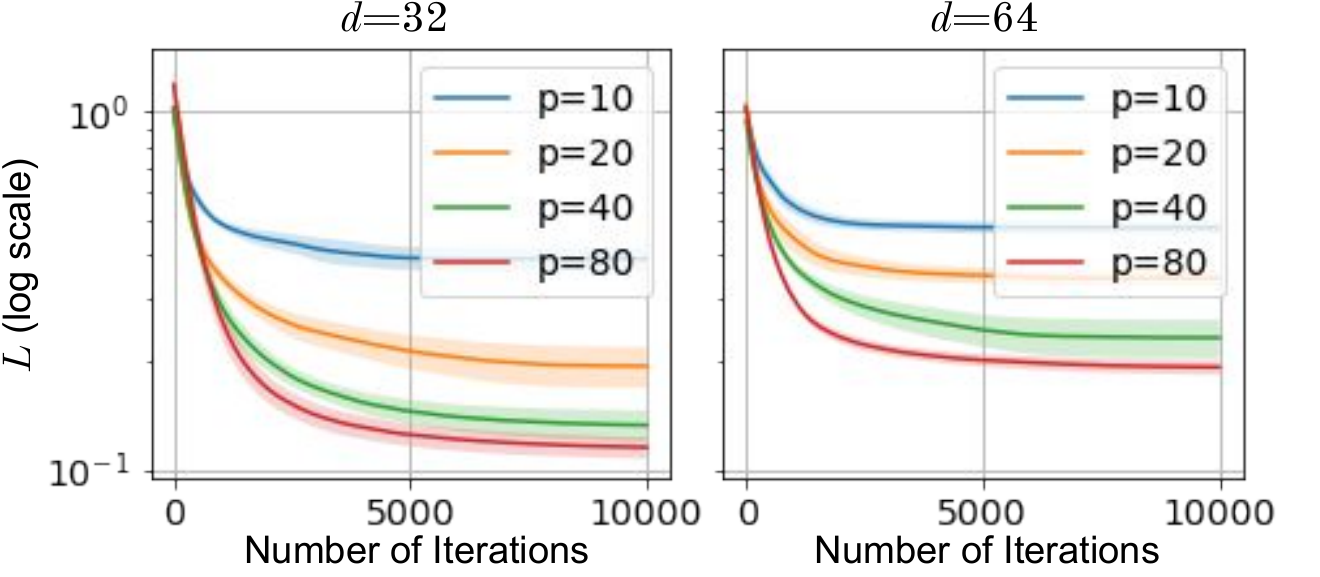}
  \caption{Sublinear convergence of QNN training. For QNNs with Pauli
    measurements for a classification task, the (log-scaled) training curves
    flatten as the number of iterations increases, indicating a sublinear
    convergence. The flattening of training curves remains for increasing
    numbers of parameters $p=10, 20, 40, 80$. The training curves are averaged
    over 10 random initialization, and the error bars are the halves of standard
    deviations.}
  \label{fig:scale1_varyp}
\end{figure}


\section{Asymptotic Dynamics of QNNs}
\label{sec:theory-fast}
As demonstrated in the previous section, the dynamics of the QNN training deviates
from the kernel regression for any choices of the number of parameters $p$ and the
dimension $d$ in the setting of Pauli measurements for classification. This
calls for a new characterization of the QNN dynamics in the regime of over-parameterization.
For a concrete definition of over-parameterization, we consider the family of the periodic
ansatze in Definition~\ref{def:partial-ansatz}, and refer to the limit of
$p\rightarrow\infty$ with a fixed generating Hamiltonian $\mlvec{H}$ as the
regime of over-parameterization.
In this section, we derive the asymptotic dynamics of QNN training when number
of parameters $p$ in the periodic ansatze goes to infinity. We start
by decomposing the dynamics of the residual $\varyr$ into a term corresponding
to the asymptotic dynamics, and a term of perturbation that vanishes as
$p\rightarrow\infty$.
As mentioned before, in the context of the gradient flow, the choice of
$\eta$ is merely a scaling of the time and therefore arbitrary.
For a QNN instance with $m$ training samples and
a $p$-parameter ansatz generated by a Hermitian $\mlvec{H}$ as defined in
Line~(\ref{eq:partial-ansatz}), we choose $\eta$ to be
$\frac{m}{p}\frac{d^{2}-1}{\tr(\mlvec{H}^{2})}$ to facilitate the presentation:
\begin{restatable}[Decomposition of the residual dynamics]{lemma}{decompresid}
  \label{lm:resid_dyn_decomp}
Let $\mathcal{S}$ be a training set with $m$ samples
$\{(\mlvec\rho_{j}, y_{j})\}_{j=1}^{m}$, and let $\mlvec{U}(\mlvec\theta)$ be
a $p$-parameter ansatz generated by a non-zero $\mlvec{H}$ as in Line~\ref{eq:partial-ansatz}.
Consider a QNN instance with a training set $\mathcal{S}$, ansatz
$\mlvec{U}(\mlvec\theta)$ and a measurement $\qnnmeasure$.  Under the
gradient flow with $\eta = \frac{m}{p}\frac{d^{2}-1}{\tr(\mlvec{H}^{2})}$, the
residual vector $\mlvec{r}(t)$ as a function of time $t$ through $\varytheta$
evolves as
\begin{align}
  \frac{d\mlvec{r}(t)}{dt} = - (\Kasym(t) + \Kpert(t)) \mlvec{r}(t)
\end{align}
where both $\Kasym$ and $\Kpert$ are functions of time through the parameterized measurement
$\varyM$, such that
\begin{align}
  (\Kasym(t))_{ij} &:= \tr\big(\imagi[\mlvec{M}(t), \mlvec{\rho}_{i}] \ \imagi[\mlvec{M}(t), \mlvec{\rho}_{j}]\big),\\
  (\Kpert(t))_{ij} &:= \tr\big(\imagi[\mlvec{M}(t), \mlvec{\rho}_{i}]
                      \otimes \imagi[\mlvec{M}(t), \mlvec{\rho}_{j}] \Delta(t)\big).
\end{align}
Here $\Delta(t)$ is a $d^{2}\times d^{2}$ Hermitian as a function of $t$ through $\mlvec\theta(t)$.
\end{restatable}
Under the random initialization by sampling $\{\mlvec{U}_{l}\}_{l=1}^{p}$ i.i.d.
from the haar measure over the special unitary group $SU(d)$,  $\Delta(0)$
concentrates at zero as $p$ increases. We further show that
$\Delta(t)-\Delta(0)$ has a bounded operator norm decreasing with number of
parameters. This allows us to associate the convergence of
the over-parameterized QNN with the properties of $\Kasym(t)$:
\begin{restatable}[Linear convergence of QNN with mean-square loss]{theorem}{qnnmselinear}
\label{thm:qnn-mse-linear}
Let $\mathcal{S}$ be a training set with $m$ samples
$\{(\mlvec\rho_{j}, y_{j})\}_{j=1}^{m}$, and let $\mlvec{U}(\mlvec\theta)$ be
a $p$-parameter ansatz generated by a non-zero $\mlvec{H}$ as in Line~(\ref{eq:partial-ansatz}).
Consider a QNN instance with the training set $\mathcal{S}$, ansatz
$\mlvec{U}(\mlvec\theta)$ and a measurement $\qnnmeasure$, trained by
gradient flow with $\eta = \frac{m}{p}\frac{d^{2}-1}{\tr(\mlvec{H}^{2})}$.
Then for sufficiently large number of parameters $p$, if the smallest eigenvalue of
$\Kasym(t)$ is greater than a constant $C_{0}$, then
with high probability over the random initialization of the periodic
ansatz, the loss function converges to zero at a linear rate
\begin{align}
  L(t) \leq L(0) \exp(-\frac{C_{0}t}{2}).
\end{align}
\end{restatable}
We defer the proof to Section~\ref{subsec:thm_qnn_mse_linear}. Similar to $\mlvec{r}(t)$, the
evolution of $\mlvec{M}(t)$ decomposes into an asymptotic term
\begin{align}
\frac{d}{dt}\mlvec{M}(t) = \sum_{j=1}^{m}r_{j}[\mlvec{M}(t), [\mlvec{M}(t),\mlvec\rho_{j}]]\label{eq:m_asymp_dynamics}
\end{align}
and a
perturbative term depending on $\Delta(t)$. Theorem~\ref{thm:qnn-mse-linear}
allows us to study the behavior of an over-parameterized QNN by simulating/characterizing the
asymptotic dynamics of $\mlvec{M}(t)$, which is significantly more accessible.

\paragraph{Application: QNN with one training sample}
To demonstrate the proposed asymptotic dynamics as a tool for analyzing
over-parameterized QNNs, we study the convergence of the QNN with one training
sample $m=1$. To set a separation from the regime of the sublinear convergence,
consider the following setting: let
$\qnnmeasure$ be a Pauli measurement, for any input
state $\mlvec{\rho}$, instead of assigning
$\hat{y}=\tr(\mlvec\rho\mlvec{U}(\mlvec\theta)^{\dagger}\qnnmeasure\mlvec{U}(\mlvec\theta))$,
take
$\gamma\tr(\mlvec\rho\mlvec{U}(\mlvec\theta)^{\dagger}\qnnmeasure\mlvec{U}(\mlvec\theta))$
as the prediction $\hat{y}$ at $\mlvec\theta$ for a scaling factor
$\gamma > 1.0$.
The $\gamma$-scaling of the measurement outcome can be viewed as a classical
processing in the context of quantum information, or as an activation function
(or a link function) in the context of machine learning, and is equivalent to a
QNN with measurement $\gamma\qnnmeasure$. The following corollary
implies the convergence of 1-sample QNN for $\gamma > 1.0$ under a mild
initial condition:
\begin{corollary}
  \label{cor:onesample}
  Let $\mlvec{\rho}$ be a $d$-dimensional pure state, and let $y$ be $\pm 1$.
  Consider a QNN instance with a Pauli measurement $\qnnmeasure$, an one-sample training set
  $\mathcal{S} = \{(\mlvec\rho, y)\}$ and an ansatz
  $\mlvec{U}(\mlvec\theta)$ defined in Line~(\ref{eq:partial-ansatz}). Assume the
  scaling factor $\gamma > 1.0$ and $p\rightarrow\infty$ with
  $\eta = \frac{d^{2}-1}{p\tr(\mlvec{H}^{2})}$. Under the initial condition that
  the prediction at $t=0$,
  $\hat{y}(0)$ is less than 1, the objective function converges linearly with
  \begin{align}
    L(t) \leq L(0) \exp(-C_{1}t)
  \end{align}
  with the convergence rate $C_{1} \geq \gamma^{2}-1$.
\end{corollary}
With a scaling factor $\gamma$ and training set
$\{(\mlvec\rho_{j}, y_{j})\}_{j=1}^{m}$, the objective function, as a
function of the parameterized measurement $\mlvec{M}(t)$, reads as:
$L(\mlvec{M}(t)) = \frac{1}{2m}\sum_{j=1}^{m}(\gamma\tr(\mlvec\rho_{j}\mlvec{M}(t)) - y_{j})^{2}$.
As stated in Theorem~\ref{thm:qnn-mse-linear}, for sufficiently large number of
parameters $p$, the convergence rate of the residual $\mlvec{r}(t)$ is
determined by $\Kasym(t)$, as the asymptotic dynamics of $\mlvec{r}(t)$ reads as
  $\frac{d}{dt}\mlvec{r} = -\Kasym(\mlvec{M}(t))\mlvec{r}(t)$ with the chosen $\eta$.
For $m = 1$, the asymptotic matrix $\Kasym$ reduces to a scalar
$k(t) = -\tr([\gamma\mlvec{M}(t),\mlvec\rho]^{2}) = 2(\gamma^{2}-\hat{y}(t)^{2})$.
$\hat{y}(t)$ approaches the label $y$ if $k(t)$ is strictly positive, which is
guaranteed for $\hat{y}(t) < \gamma$. Therefore $|\hat{y}(0)|<1$ implies that
$|\hat{y}(t)|<1$ and $k(t)\geq 2(\gamma^{2}-1)$ for all $t>0$.

In Figure~\ref{fig:onesample} (top), we plot the training curves of one-sample QNNs with
$p=320$ and
varying $\gamma = 1.2, 1.4, 2.0, 4.0, 8.0$ with the same learning rate
$\eta=1e-3/p$. As predicted in Corollary~\ref{cor:onesample}, the rate of convergence increases with the
scaling factor $\gamma$.
The proof of the corollary additionally implies that $k(t)$ depends on
$\hat{y}(t)$: the convergence rate changes over time as the prediction $\hat{y}$
changes. Therefore, despite the linear convergence, the dynamics is different from that of kernel
regression, where the kernel remains constant during training in the limit $p\rightarrow\infty$.

In Figure~\ref{fig:onesample} (bottom), we plot the empirical rate of convergence
$-\frac{d}{dt}\ln L(t)$ against the rate predicted by $\hat{y}$. Each data
point is calculated for QNNs with different $\gamma$ at different time steps by
differentiating the logarithms of the training curves. The scatter plot displays
an approximately linear dependency, indicating the proposed asymptotic dynamics is
capable of predicting how the convergence rate changes during training, which is
beyond the explanatory power of the kernel regression model. Note that the slope
of the linear relation is not exactly one. This is because we choose a learning rate much
smaller than $\eta$ in the corollary statement to simulate the dynamics of gradient flow.

\begin{figure}[!htbp]
  \centering
  \includegraphics[width=0.65\linewidth]{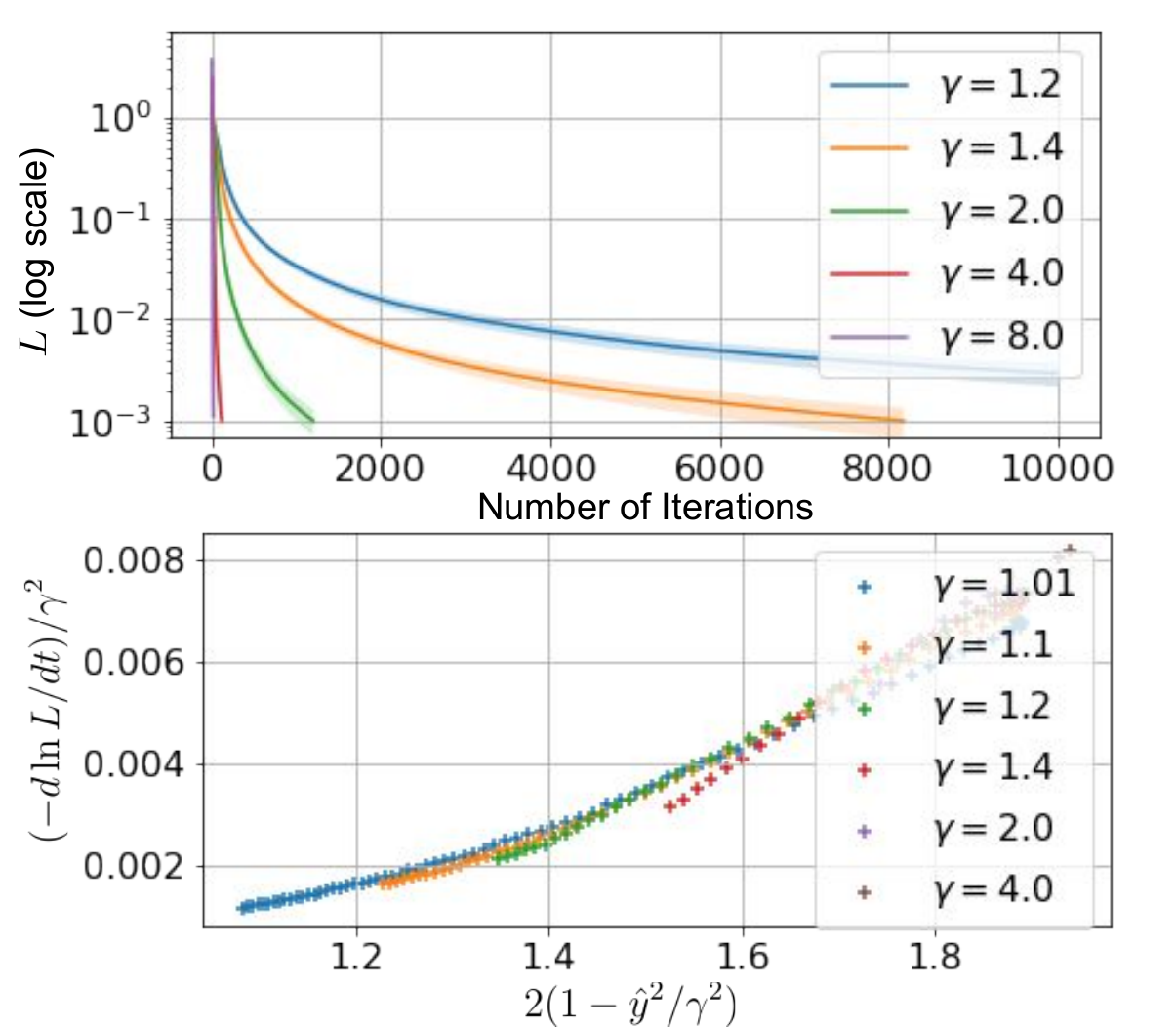}
  \caption{
    (Top) The training curves of one-sample QNNs with varying $\gamma$. The
    smallest convergence rate $-d\ln L / dt$ during training (i.e. the slope of
    the training curves under the log scale) increases with $\gamma$.
    (Bottom) The convergence rate $-d\ln L/dt|_{t=T}$ as a function of
    $2(\gamma^{2}-\hat{y}^{2}(T))$ (jointly scaled by
    $1/\gamma^{2}$ for visualization) are
    evaluated at different time steps $T$ for different $\gamma$. The approximately
    linear dependency shows that the proposed dynamics captures the QNN
    convergence beyond the explanatory power of the kernel regressions.
  }
  \label{fig:onesample}
\end{figure}
QNNs with one training sample have been considered before (e.g.
\cite{liu2022analytic}), where the linear convergence has been shown under the assumption of ``frozen QNTK", namely assuming $\mlvec{K}$,
the time derivative of the log residual remains almost constant throughout training. In the corollary above, we provide an
end-to-end proof for the one-sample linear convergence without assuming a frozen
$\mathbf{K}$. In fact, we observe that in our setting
$\mathbf{K} = 2(\gamma^2 - \hat{y}(t))$ changes with $\hat{y}(t)$ (see also
Figure~\ref{fig:onesample}) and is therefore not frozen.

\paragraph{QNN convergence for $m>1$}
To characterize the convergence of QNNs with $m>1$, we
seek to empirically study the asymptotic dynamics in Line~(\ref{eq:m_asymp_dynamics}). According to
Theorem~\ref{thm:qnn-mse-linear}, the (linear) rate of convergence is
lower-bounded by the smallest eigenvalue of $\Kasym(t)$, up to an constant scaling. In
Figure~\ref{fig:smallest_eig}, we simulate the asymptotic dynamics with various
combinations of $(\gamma, d, m)$, and evaluate the smallest eigenvalue of
$\Kasym(t)$ throughout the dynamics (Figure~\ref{fig:smallest_eig}, details
deferred to Section~\ref{sec:app_exp}). For sufficiently large dimension $d$,
the smallest eigenvalue of $\Kasym$ depends on the ratio between the number of
samples and the system dimension $m/d$ and is proportional to the square of the
scaling factor $\gamma^{2}$.

\begin{figure}[!htbp]
  \centering
  \includegraphics[width=.7\linewidth]{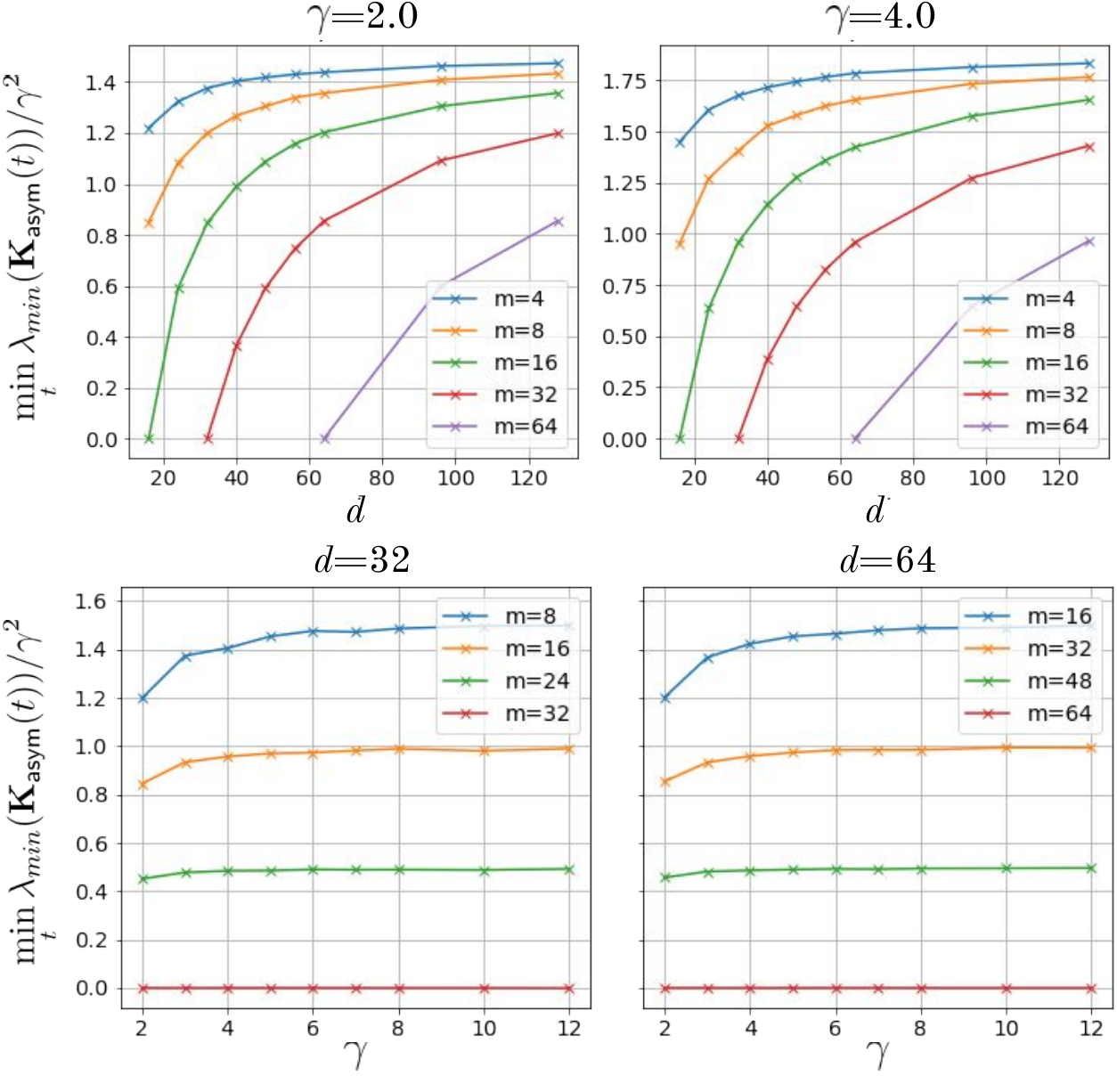}
  \caption{
    The smallest eigenvalue of $\Kasym$ for the asymptotic dynamics with
    varying system dimension $d$, scaling factor $\gamma$ and number of
    training samples $m$. For sufficiently large $d$, the smallest eigenvalue
    depends on the ratio $m/d$ and is proportional to the square of the scaling factor $\gamma^{2}$.
  }
  \label{fig:smallest_eig}
\end{figure}

Empirically, we observe that the smallest convergence rates for training QNNs
are obtained near the global minima (See
Figure~\ref{fig:Kchg_vary_qnn} in the appendix), suggesting the bottleneck of convergence occurs when $L$ is small. 

We now give theoretical evidence that, at most of the global minima, the eigenvalues of $\Kasym$ are lower bounded by $2\gamma^2(1-1/\gamma^2 - O(m^2/d))$, suggesting a linear convergence in the neighborhood of these minima. 
To make this notion precise, we define the uniform measure over global minima as follows:
consider a set of pure input states $\{\mlvec\rho_{j}=\mlvec{v}_{j}\mlvec{v}_{j}^{\dagger}\}_{j=1}^{m}$ that are  mutually
orthogonal (i.e. $\mlvec{v}_{i}^{\dagger}\mlvec{v}_{j}=0$ if $i\neq j$).
For a large dimension $d$, the global minima of the asymptotic dynamics is achieved when the
objective function is $0$. Let $\mlvec{u}_{j}(t)$ (resp. $\mlvec{w}_{j}(t)$)
denote the components of $\mlvec{v}_{j}$ projected to the positive (resp. negative)
subspace of the measurement $\mlvec{M}(t)$ at the global minima. Recall that for a $\gamma$-scaled QNN with a
Pauli measurement, the predictions $\hat{y}(t) = \gamma\tr(\rho_{j}\mlvec{M}(t)) = \gamma(\mlvec{u}_{j}^{\dagger}(t)\mlvec{u}_{j}(t)-\mlvec{w}_{j}^{\dagger}(t)\mlvec{w}_{j}(t))$.
At the global minima, we have $\mlvec{u}_{j}(t) = \frac{1}{2}(1\pm 1/\gamma)\hat{\mlvec{u}}_{j}(t)$ for some unit
vector $\hat{\mlvec{u}}_{j}(t)$ for the $j$-th training sample with label $\pm 1$. On the other hand, given a set of unit vectors
$\{\hat{\mlvec{u}}_{j}\}_{j=1}^{m}$ in the positive subspace, there is a corresponding set of $\{\mlvec{u}_{j}(t)\}_{j=1}^{m}$ and
$\{\mlvec{w}_{j}(t)\}_{j=1}^{m}$ such that $L=0$ for sufficiently large $d$. By uniformly and independently
sampling a set of unit vectors $\{\hat{\mlvec{u}}_{j}\}_{j=1}^{m}$ from the
$d/2$-dimensional subspace associated with the positive eigenvalues of
$\mlvec{M}(t)$, we induce a uniform distribution over all the global minima.
The next theorem characterizes $\Kasym$ under such an induced uniform distribution over all the global minima:
\begin{restatable}{theorem}{asympeig}
\label{thm:smallest_eig}
Let $\mathcal{S}=\{(\mlvec\rho_{j}, y_{j})\}_{j=1}^{m}$ be a training set with
orthogonal pure states $\{\mlvec\rho_{j}\}_{j=1}^{m}$ and equal number of
positive and negative labels $y_{j}\in\{\pm 1\}$.
Consider the smallest eigenvalue $\lambda_{g}$ of $\Kasym$ at the global minima of the asymptotic dynamics of an over-parameterized
QNN with the training set $\mathcal{S}$, scaling factor $\gamma$ and system dimension $d$. With probability
$\geq 1 - \delta$ over the uniform
measure over all the global minima
\begin{align}
  \lambda_{g}\geq 2\gamma^{2}(1-\frac{1}{\gamma^{2}}-C_{2}\max\{\frac{m^{2}}{d}, \frac{m}{d}\log\frac{2}{\delta}\}),
\end{align}
which is strictly positive for large $\gamma > 1$ and $d=\Omega(\mathsf{poly}(m))$.
Here $C_{2}$ is a positive constant.
\end{restatable}
We defer the proof of Theorem~\ref{thm:smallest_eig} to Section~\ref{sec:app_globalminima} in the appendix. A similar notion of a uniform measure over global minima was also used in \citet{canatar2021spectral}.
Notice that the uniformness is dependent on the parameterization of the global minima, and the uniform measure over all the global minima is not necessarily the
measure induced by random initialization and gradient-based training. Therefore
Theorem~\ref{thm:smallest_eig} is not a rigorous depiction of the distribution of convergence rate for a randomly-initialized over-parameterized QNN.
Yet the prediction of the theorem aligns well with the empirical observations in Figure~\ref{fig:smallest_eig} and
suggests that by scaling the QNN measurements, a faster convergence can be achieved:
In Figure~\ref{fig:scale2_varyp}, we simulate $p$-parameter QNNs with dimension $d = 32$ and
$64$ with a scaling factor $\gamma = 4.0$ using the same setup as in
Figure~\ref{fig:scale1_varyp}. The training early stops  when the average $L(t)$ over the random seeds is less than $1\times 10^{-2}$.
In contrast to Figure~\ref{fig:scale1_varyp}, the convergence rate $-d\ln L/dt$
does not vanish as $L\rightarrow 0$, suggesting a simple (constant) scaling of the measurement outcome can lead to convergence within much fewer number of iterations. 
\begin{figure}[!htbp]
  \centering
  \includegraphics[width=.7\linewidth]{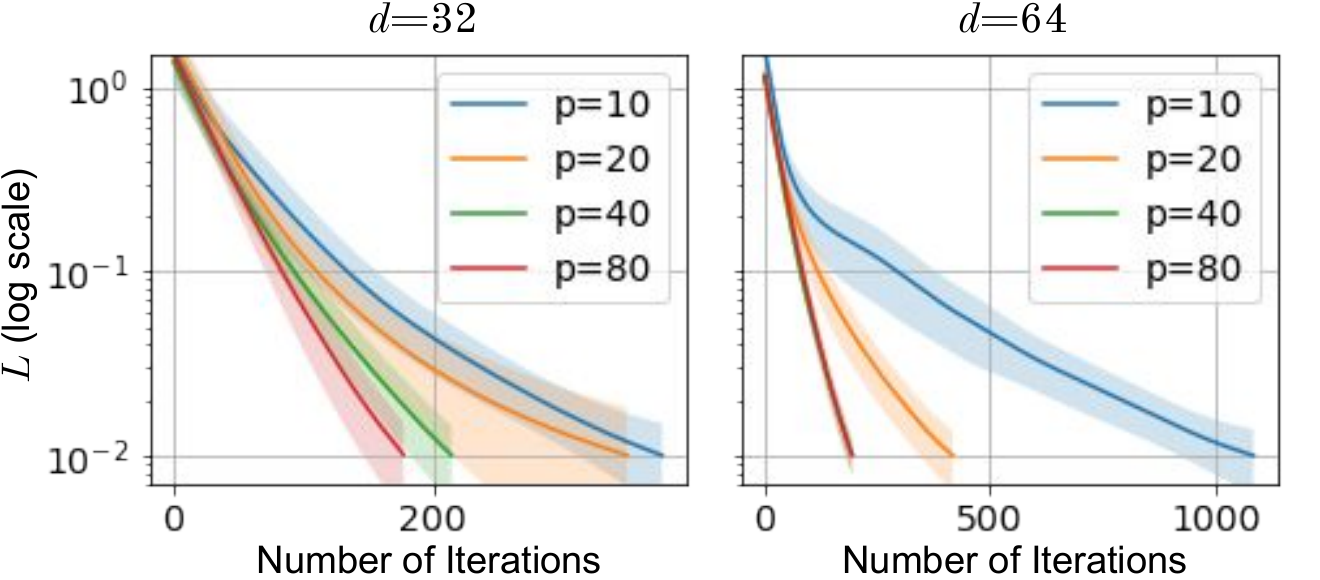}
  \caption{
    Training curves of QNNs with $\gamma=4.0$ for learning a $4$-sample dataset
    with labels $\pm 1$. For $p=10, 20, 40, 80$, the rate of convergence is
    greater than $0$ as $L\rightarrow 0$, and it takes less than $1000$
    iterations for $L$ in most of the instances to convergence below
    $1\times 10^{-2}$. In contrast, in Figure~\ref{fig:scale1_varyp},
    $L > 1\times 10^{-1}$ after $10000$ iterations despite the increasing number
    of parameters.
  }
  \label{fig:scale2_varyp}
\end{figure}

Another implication of Theorem~\ref{thm:smallest_eig} is the deviation of QNN dynamics from any kernel regressions. By straight-forward calculation, the normalized matrix $\Kasym(0) / \gamma^2$  at the random initialization is independent of the choices of $\gamma$. In contrast, the typical value of $\lambda_g / \gamma^2$ in Theorem~\ref{thm:smallest_eig} is dependent on $\gamma^2$, suggesting non-negligible changes in the matrix $\Kasym(t)$ governing the dynamics of $\mlvec{r}$ for finite scaling factors $\gamma$. Such phenomenon is empirically verified in Figure~\ref{fig:reKchg_vary} in the appendix.


\section{Limitations and Outlook}
\label{sec:conclusion}
In the setting of $m>1$, the proof of the linear convergence of QNN
training (Section~\ref{sec:theory-fast}) relies on the convergence of the
asymptotic QNN dynamics as a premise. Given our empirical results, an interesting future direction might be to rigorously characterize
the condition for the convergence of the asymptotic dynamics. Also we mainly consider (variants of) two-outcome measurements $\qnnmeasure$ with two
eigensubspaces. It might be interesting to look into measurements with more
complicated spectrums and see how the shapes of the spectrums affect the rates
of convergence.

A QNN for learning a classical dataset is composed of three parts: a classical-to-quantum
encoder, a quantum classifier and a readout measurement. Here we have mainly
focused on the stage after encoding, i.e. training a QNN classifier to
manipulate the density matrices containing classical information that are
potentially too costly for a classically-implemented linear model. Our analysis
highlights the necessity for measurement design, assuming the design of the
quantum classifier mixes to the full $d\times d$ special unitary group. Our
result can be combined with existing techniques of classifier designs (i.e.
ansatz design) (\cite{ragone2022representation, larocca2021theory, wang2022symmetric,ourvqe2022}) by
engineering the invariant subspaces, or be combined with encoder designs
explored in \citep{huang2021power,du2022demystify}.

\section*{Acknowledgements}
We thank E. Anschuetz, B. T. Kiani, J. Liu and anonymous reviewers for useful comments. This work received support from the U.S. Department of Energy, Office
of Science, Office of Advanced Scientific Computing Research, Accelerated Research in Quantum Computing and
Quantum Algorithms Team programs, as well as the U.S.
National Science Foundation grant CCF-1816695, and CCF-1942837 (CAREER).

\section*{Disclaimer}
This paper was prepared with synthetic data and for informational purposes by the teams of researchers from the various institutions identified above, including the Global Technology Applied Research Center of JPMorgan Chase \& Co. This paper is not a product of the JPMorgan Chase Institute.  Neither JPMorgan Chase \& Co. nor any of its affiliates make any explicit or implied representation or warranty and none of them accept any liability in connection with this paper, including, but limited to, the completeness, accuracy, reliability of information contained herein and the potential legal, compliance, tax or accounting effects thereof.  This document is not intended as investment research or investment advice, or a recommendation, offer or solicitation for the purchase or sale of any security, financial instrument, financial product or service, or to be used in any way for evaluating the merits of participating in any transaction.

\bibliography{references}
\bibliographystyle{abbrvnat}

\newpage
\appendix
\section{Proofs for Section~\ref{sec:qnn-dynamics}}
\label{sec:app-proofqnndynamics}
\subsection{Proof of Lemma~\ref{lm:qnndynamics}}
\label{subsec:qnndynamics_proof}

\qnndynamics*
\begin{proof}
  For succinctness, we drop the dependency on $\varytheta$ when there are no ambiguity.
    The unitary $\mlvec{U}_{r:p}(\mlvec\theta)$  depends on $\mlvec\theta_{l}$
    for $l\geq r$:
  \begin{align}
      \frac{\partial \mlvec{U}_{r:p}}{\partial\theta_{l}}
    = \mlvec{U}_{l:p}(\mlvec{\theta}) (-\compI\generatorH{l}) \mlvec{U}_{r:l-1}(\mlvec\theta)
    = -\compI\mlvec{U}_{l:p}\generatorH{l}\mlvec{U}_{l:p}^{\dagger}\mlvec{U}_{r:p}.
  \end{align}
  Therefore for all $l \in [p]$
  \begin{align}
    \frac{\partial \mlvec{M}(\mlvec{\theta}(t))}{\partial \theta_{l}}
    &=
    \mlvec{U}_{0}^{\dagger}
      \left(\partialtheta{\mlvec{U}_{1:p}}{l} \right)^{\dagger}\qnnmeasure\mlvec{U}_{{1:p}}\mlvec{U}_{0} + \mlvec{U}_{0}^{\dagger}\mlvec{U}_{1:p}^{\dagger}\qnnmeasure\partialtheta{\mlvec{U}_{{1:p}}}{l}\mlvec{U}_{0}, \\
    &=
      \compI(\mlvec{U}_{0}^{\dagger}\mlvec{U}_{1:p}^{\dagger}\mlvec{U}_{l:p}\generatorH{l}\mlvec{U}_{l:p}^{\dagger}\qnnmeasure \mlvec{U}_{1:p}\mlvec{U}_{0}) - \compI(\mlvec{U}_{0}^{\dagger}\mlvec{U}_{1:p}^{\dagger}\qnnmeasure \mlvec{U}_{l:p}\generatorH{l}\mlvec{U}_{l:p}^{\dagger}\mlvec{U}_{1:p}\mlvec{U}_{0}), \\
    &= \compI[\tildeH{l},\mlvec{M}(\mlvec{\theta}(t))].
  \end{align}
  By the chain rule with matrix parameters, we have
  \begin{align}
    \partialtheta{L(\mlvec{\theta}(t))}{l} = \tr\big(\nabla_{\mlvec{M}}L \partialtheta{\mlvec{M}}{l} \big) = \compI\tr(\nabla_{\mlvec{M}}L[\tildeH{l},\mlvec{M}(\mlvec{\theta}(t))]).
  \end{align}
  Furthermore, due to the gradient flow dynamics,
  \begin{align}
    \diffT{\paramM} &= \sum_{l=1}^{p}\diffT{\theta_{l}}\partialtheta{\mlvec{M}(\mlvec{\theta}(t))}{l} = - \eta\sum_{l=1}^{p}\partialtheta{L(\mlvec{\theta}(t))}{l} \partialtheta{\mlvec{M}(\mlvec{\theta}(t))}{l},\\
    &= \eta \sum_{l=1}^{p}\tr(\nabla_{\mlvec{M}}L[\tildeH{l},\mlvec{M}(\mlvec{\theta}(t))])[\tildeH{l},\mlvec{M}(\mlvec{\theta}(t))].
  \end{align}
By plugging in
$\nabla_{\mlvec{M}}L = -\frac{1}{m}\sum_{j=1}^{m}r_{j}\mlvec{\rho}_{j}$, we show
that the parameterized measurement
$\mlvec{M}(\mlvec\theta)=\mlvec{U}^{\dagger}(\mlvec\theta)\qnnmeasure\mlvec{U}(\mlvec\theta)$
follows the dynamics
\begin{align}
d\mlvec{M}(\mlvec\theta(t))/dt=
\frac{\eta}{m} \sum_{l=1}^{p}\tr\big(\sum_{j=1}^{m}r_{j}\mlvec\rho_{j} \compI[\tildeH{l}, \mlvec{M}(\mlvec{\theta}(t))]\big) \compI[\tildeH{l}, \mlvec{M}(\mlvec{\theta}(t))].
\end{align}
By definition $r_{i}:=y_{i}-\hat{y}_{i}$, and
    \begin{align}
      \diffT{r_{i}} &= -\diffT{\tr(\mlvec\rho_{i}\mlvec{M}(\mlvec{\theta}(t)))} = -\tr\big(\mlvec\rho_{i}\diffT{\paramM}\big)\\
                    &=  -\frac{\eta}{m} \sum_{l=1}^{p}\tr\big(\sum_{j=1}^{m}r_{j}\mlvec\rho_{j} \compI[\tildeH{l}, \mlvec{M}(\mlvec{\theta}(t))]\big) \tr\big(\mlvec{\rho_{i}}\compI[\tildeH{l}, \mlvec{M}(\mlvec{\theta}(t))]\big)\\
                    &= -\frac{\eta}{m} \sum_{j=1}^{m} r_{j} \big(\tr\big(\mlvec{\rho}_{i}\compI[\tildeH{l}, \mlvec{M}(\mlvec{\theta}(t))]\big)\tr\big(\mlvec{\rho}_{j}\compI[\tildeH{l}, \mlvec{M}(\mlvec{\theta}(t))]\big)\big)\\
      &= -\frac{\eta}{m} \sum_{j=1}^{m} r_{j} \big(\tr\big(\compI[\paramM,\mlvec{\rho}_{i}]\tildeH{l}\big)\tr\big(\compI[\paramM,\mlvec{\rho}_{j}]\tildeH{l}\big)\big).
    \end{align}
    The last equality is due to the cyclicity of the trace operation. Making the
    identification
    $K_{ij}(\paramM) = \big(\tr\big(\compI[\paramM,\mlvec{\rho}_{i}]\tildeH{l}\big)\tr\big(\compI[\paramM,\mlvec{\rho}_{j}]\tildeH{l}\big)\big)$,
    we have
    \begin{align}
    \frac{d\varyr}{dt}=-\frac{\eta}{m}\mlvec{K}(\varyM)\varyr.
    \end{align}
  \end{proof}

\subsection{Proof of Theorem~\ref{thm:sublinear-convergence}}
\label{subsec:slow_proof}

\begin{proof}
The mean squared loss function $L(\varytheta)$ can be expressed as
$\frac{1}{2m}\mlvec{r}(\varytheta)^{T}\mlvec{r}(\varytheta)$. Using
Lemma~\ref{lm:qnndynamics}, the rate of convergence can be lower-bounded as
\begin{align}
  &\frac{1}{\varyl}\frac{d\varyl}{dt}\\
  =&\frac{1}{\varyr^{T}\varyr}\frac{d}{dt}\varyr^{T}\varyr, \\
  =&-\frac{2\eta}{m}\cdot\frac{\varyr^{T}\mlvec{K}(\varytheta)\varyr}{\varyr^{T}\varyr}, \\
  \geq& -\frac{2\eta}{m} \lambda_{\max}(\mlvec{K}(\varytheta)).
\end{align}
The positive semi-definiteness of $\varyK$ suggests that
$\lambda_{\max}(\varyK) \leq \tr(\varyK)$. We now proceed to bound
$\tr(\varyK)$. Since the eigenvalues of $\qnnmeasure$ and $\mlvec{M}(\mlvec\theta)$ all lie
in $\{\pm 1\}$, $\mlvec{M}(\varytheta)$ decomposes into the difference of to
projections, $\mlvec{\Pi}_{+}(\varytheta)$ and $\mlvec{\Pi}_{-}(\varytheta)$,
projecting onto the subspaces associated with eigenvalues of $+1$ and $-1$
respectively. When $\hat{y}_{j}$ approaches $y_{j}$, the input state
$\mlvec{\rho}_{j}$ lies almost completely in one of the eigen-subspaces, leading to a
vanishing commutator $\compI[\varyM, \mlvec{\rho}_{j}]$ such that
$K_{jj}(\varytheta)$ approaches zero:

Let $\mlvec{v}_{j}$ be the statevector representation of the pure state
$\mlvec{\rho}_{j}$, such that
$\mlvec{\rho}_{j}=\mlvec{v}_{j}\mlvec{v}_{j}^{\dagger}$. Vector $\mlvec{v}_{j}$
decomposes into the components within the positive and negative
eigen-subspaces of $\varyM$:
$\mlvec{v}_{j} = \mlvec{u}_{j}(\varytheta) + \mlvec{w}_{j}(\varytheta)$, where
$\mlvec{u}_{j}(\varytheta) = \mlvec{\Pi}_{+}(\varytheta)\mlvec{v}_{j}$ and $\mlvec{w}_{j}(\varytheta) = \mlvec{\Pi}_{-}(\varytheta)\mlvec{v}_{j}$.
In the following we omit the arguments $\varytheta$ in $\mlvec{u}_{j}$ and $\mlvec{v}_{j}$ for succinctness, but the time dependence is to be implicitly understood.
The commutator between the parameterized measurement and the input state can be
written as
$[\varyM, \mlvec{\rho}_{j}] = 2(\mlvec{u}_{j}\mlvec{w}_{j}^{\dagger}-\mlvec{w}_{j}\mlvec{u}_{j}^{\dagger})$. Therefore
  \begin{align}
    |\tr(\compI[\mlvec{M},\mlvec{\rho}_{j}]\tildeH{l})| \le 4\opnorm{\tildeH{l}}\norm{\mlvec{u}_{j}}\norm{\mlvec{w}_{j}}.
  \end{align}
  Assume without loss of generality that the $j$-{th} label $y_{j}$ is $+1$. Then
  $\norm{\mlvec{u}_{j}}^{2} + \norm{\mlvec{w}_{j}}^{2} = \norm{\mlvec{v}_{j}}^{2} = 1$
  by definition, and
  $\norm{\mlvec{u}_{j}}^{2} - \norm{\mlvec{w}_{j}}^{2} = \tr(\varyM\mlvec{\rho}_{j}) = y_{j} - r_{j} = 1 - r_{j}$.
  Then
  $\norm{\mlvec{w}_{j}}^{2} = |r_{j}|/2$, and
  $\norm{\mlvec{u}_{j}}^{2}\norm{\mlvec{w}_{j}}^{2} = (1-|r_{j}|/2)|r_{j}|/2$.

  Therefore we have,
  \begin{align}
    K_{jj}(\mlvec\theta(t))
    =&\sum_{l=1}^{p}\tr^{2}(\compI[\varyM, \mlvec{\rho}_{j}]\mlvec{H}_{l})\\
    \leq& 16\sum_{l=1}^{p}\opnorm{\mlvec{H}_{l}}^{2}\frac{|r_{j}|}{2}(1-\frac{|r_{j}|}{2})\\
    \leq& 16\sum_{l=1}^{p}\opnorm{\mlvec{H}_{l}}^{2}\frac{|r_{j}|}{2}
  \end{align}
  As a result
  \begin{align}
    & \frac{1}{L(\varytheta)}\frac{d\varyl}{dt}\\
    & \geq -\frac{2\eta}{m} \tr(\mlvec{K}(\varytheta))\ge -\frac{2\eta}{m} \sum_{i=1}^{m}K_{ii}\\
    & \geq -\frac{16\eta}{m}\sum_{l=1}^{p}\opnorm{\mlvec{H}_{l}}^{2}\sum_{i=1}^{m}|r_{j}|\\
    & \ge - 16\sqrt{2} \eta \sum_{l=1}^{p}\opnorm{\mlvec{H}_{l}}^{2} \sqrt{L(\varytheta)}\\
    & = - 16\sqrt{2} \eta \sum_{l=1}^{p}\opnorm{\generatorH{l}}^{2} \sqrt{L(\varytheta)}.
  \end{align}
  Here we use the fact that
  $\sum_{j=1}^{m}|r_{j}|\leq \sqrt{m} \sqrt{\varyr^{T}\varyr} = \sqrt{2m^{2}\varyl}$.

  The theorem statement follows directly by integrating the inequality above:
  \begin{align}
    &\varyl^{-\frac{3}{2}}d\varyl \geq -24 \eta\sum_{l=1}^{p}\opnorm{\generatorH{l}}^{2} dt\\
    \implies& -2 d(\varyl^{-1/2}) \geq -24\eta\sum_{l=1}^{p}\opnorm{\generatorH{l}}^{2}dt\\
    \implies& L(\mlvec\theta(T))^{-\frac{1}{2}}-L(\mlvec\theta(0))^{-\frac{1}{2}} \leq 12\eta\sum_{l=1}^{p}\opnorm{\generatorH{l}}^{2}T\\
    \implies& L(\mlvec\theta(T))^{-\frac{1}{2}}-c_{0}\leq c_{1}T
  \end{align}
\end{proof}
Note that the same ``at most sublinear convergence'' holds for a measurement
$\mathbf{M}_{0}$ such that $\mathbf{M}_0 = \boldsymbol{\Pi}_+-\boldsymbol{\Pi}_-$ and $\boldsymbol{\Pi}_+ + \boldsymbol{\Pi}_- + \boldsymbol{\Pi}_0 = \mathbf{I}$ for some non-zero projection $\boldsymbol{\Pi}_0$. The proof still holds with the following modification: define $s_j:=\|\mathbf{u}_j\|^2 + \|\mathbf{w}_j\|^2 \leq 1$, we have
 \begin{align*}
   &\|\mathbf{u}_j\|^2\cdot\|\mathbf{w}_j\|^2\\
   =& (\frac{s_j-y_j+r_j}{2})(\frac{s_j+y_j-r_j}{2})\\
   =& \frac{s_j^2-(y_j-r_j)^2}{4}\\
   \leq& \frac{1-(y_j-r_j)^2}{4}\\
   =& \frac{(1-y_j)^2+2y_jr_j - r^2_j}{4}\\
   =& \frac{y_jr_j}{2} - \frac{r^2_j}{4}\leq \frac{y_jr_j}{2}\\
   = & \frac{|r_j|}{2}
 \end{align*}
 The last equality follows from the fact that  $r_j \geq 0$ (resp. $r_j\leq 0$) for $y_j=1$ (resp. $y_j=-1$).


\section{Proofs for the asymptotic dynamics}
\label{sec:app_prooffast}
\subsection{Proof of Lemma~\ref{lm:resid_dyn_decomp}}
\label{subsec:lm_resid_dyn_decomp}
\decompresid*
Throughout the proof, we make use of the following notations. Let $\mathcal{H}$ be a $d$-dimensional Hilbert space, and let $\{\mlvec{e}_a\}_{a\in[d]}$ be a basis of $\mathcal{H}$. We use $\mlvec{I}_{d\times d}$ denote the identity matrix $\sum_{a\in[d]}\mlvec{e}_a\mlvec{e}_a^\dagger$.
We use $\otimes$ for kronecker products on vectors, matrices and Hilbert spaces.
For the $d^2\times d^2$-dimensional product space $\mathcal{H}\otimes \mathcal{H}$, let $\mlvec{W}_{d^2\times d^2}$ denote the swap matrix $\sum_{a,b\in[d]}\mlvec{e}_a\mlvec{e}_b^\dagger\otimes \mlvec{e}_b\mlvec{e}_a^\dagger$. 

We will also make use of the well-known integration formula with respect to the haar measure over $d$-dimensional unitaries (see e.g. \citet{collins2006integration} for more details). 
\begin{proof}
As proven in Lemma~\ref{lm:qnndynamics}, we track the
dynamics of the parameterized measurement $\mlvec{M}(\mlvec\theta)$:
\begin{align}
  \frac{d\mlvec{M}(\mlvec\theta)}{dt}
  = &\sum_{l=1}^{p}\frac{d\theta_{l}}{dt} \cdot \frac{\partial \mlvec{M}(\mlvec\theta)}{\partial \theta_{l}}\\
  = &\sum_{l=1}^{p}(-\eta)\tr\big(\imagi[\mlvec{H}_{l}, \mlvec{M}(\mlvec\theta)] \nabla_{\mlvec{M}}L \big) \imagi [\mlvec{H}_{l}, \mlvec{M}(\mlvec\theta)]\\
  = &\sum_{l=1}^{p}\eta\tr\big(\imagi[\nabla_{\mlvec{M}}L, \mlvec{M}(\mlvec\theta)] \mlvec{H}_{l} \big) \imagi [\mlvec{H}_{l}, \mlvec{M}(\mlvec\theta)]\\
  = &\sum_{l=1}^{p}\eta \imagi [\tr\big(\imagi[\nabla_{\mlvec{M}}L, \mlvec{M}(\mlvec\theta)] \mlvec{H}_{l} \big)\mlvec{H}_{l}, \mlvec{M}(\mlvec\theta)]\\
  = &\sum_{l=1}^{p}\eta \imagi [\tr_{1}\big((\imagi[\nabla_{\mlvec{M}}L, \mlvec{M}(\mlvec\theta)]\otimes\mlvec{I}) (\mlvec{H}_{l}\otimes\mlvec{H}_{l})\big), \mlvec{M}(\mlvec\theta)].
\end{align}
Here $\tr_{{1}}(\cdot)$ is the partial trace: Given the product of two Hilbert spaces $\mathcal{H}_1 \otimes \mathcal{H}_2$, the partial trace on the first Hilbert space is a linear mapping such that
\begin{align*}
    \tr_1\big(\mathbf{A}\otimes \mathbf{B}\big) = \tr(\mathbf{A})\mathbf{B}
\end{align*}
for any Hermitians $\mathbf{A}$ and $\mathbf{B}$ on the spaces $\mathcal{H}_1$ and $\mathcal{H}_2$. By linearity,
  \begin{align*}
    \tr_1\big(\sum_l\mathbf{A}_l\otimes \mathbf{B}_l\big) = \sum_l\tr(\mathbf{A}_l)\mathbf{B}_l
\end{align*}
for any Hermitians $\{\mathbf{A}_l\}$ and $\{\mathbf{B}_l\}$ on the spaces $\mathcal{H}_1$ and $\mathcal{H}_2$.

Let $Z(\mlvec{H},d)$ denote the ratio $\frac{\tr(\mlvec{H}^{2})}{d^{2}-1}$, the
learning rate $\eta$ can be expressed as $\frac{m}{pZ(\mlvec{H},d)}$. Let $\mlvec{Y}(\mlvec\theta(t))$ denote the normalized
  $d^{2}\times d^{2}$-complex matrix
  $\frac{1}{pZ(\mlvec{H},d)}\sum_{l=1}^{p}\mlvec{H}_{l}\otimes\mlvec{H}_{l}$ for $\mlvec{H}_l$ defined in Lemma~\ref{lm:qnndynamics}
  and let $\mlvec{Y}^{\star}$ denote
  $\mlvec{W}_{d^{2}\times d^{2}} - \frac{1}{d} \mlvec{I}_{d^{2}\times d^{2}}$,
  the asymptotic version of $\mlvec{Y}$. We can accordingly decompose the dynamics into the
  asymptotic dynamics and the deviation (perturbation) from the asymptotic
  dynamics:
  \begin{align}
    \frac{d\mlvec{M}(\mlvec\theta)}{dt}
    = &(\eta p Z(\mlvec{H},d)) \imagi [\tr_{1}\big((\imagi[\nabla_{\mlvec{M}}L, \mlvec{M}(\mlvec\theta)]\otimes\mlvec{I}) \mlvec{Y}\big), \mlvec{M}(\mlvec\theta)]\\
    =
    & (\eta p Z(\mlvec{H},d)) \imagi [\tr_{1}\big((\imagi[\nabla_{\mlvec{M}}L, \mlvec{M}(\mlvec\theta)]\otimes\mlvec{I}) \mlvec{Y}^{\star}\big), \mlvec{M}(\mlvec\theta)]\\
    & + (\eta p Z(\mlvec{H},d)) \imagi [\tr_{1}\big((\imagi[\nabla_{\mlvec{M}}L, \mlvec{M}(\mlvec\theta)]\otimes\mlvec{I}) (\mlvec{Y}(\mlvec\theta(t)) - \mlvec{Y}^{\star})\big), \mlvec{M}(\mlvec\theta)]\\
    =
    & (\eta p Z(\mlvec{H},d)) \imagi [(\imagi[\nabla_{\mlvec{M}}L, \mlvec{M}(\mlvec\theta)], \mlvec{M}(\mlvec\theta)]\\
    & + (\eta p Z(\mlvec{H},d)) \imagi [\tr_{1}\big((\imagi[\nabla_{\mlvec{M}}L, \mlvec{M}(\mlvec\theta)]\otimes\mlvec{I}) (\mlvec{Y}(\mlvec\theta(t)) - \mlvec{Y}^{\star})\big), \mlvec{M}(\mlvec\theta)]\\
    =
      & - (\eta p Z(\mlvec{H},d)) [\mlvec{M}(\mlvec\theta), [\mlvec{M}(\mlvec\theta), \nabla_{\mlvec{M}}L]]\\
      & -(\eta p Z(\mlvec{H},d)) [\mlvec{M}(\mlvec\theta), \tr_{1}\big(([\mlvec{M}(\mlvec\theta), \nabla_{\mlvec{M}}L]\otimes\mlvec{I}) (\mlvec{Y}(\mlvec\theta(t)) - \mlvec{Y}^{\star})\big)]
  \end{align}
  Plugging in that
  $\nabla_{\mlvec{M}}L(\mlvec{M}(\mlvec\theta)) = -\frac{1}{m}\sum_{i=1}^{m}r_{i}\mlvec\rho_{i}$
  with the residual
  $r_{i}:=y_{i}-\hat{y}_{i} = \tr(\mlvec{M}(\mlvec\theta)\rho_{i}) - y_{i}$:
  \begin{align}
    \frac{d\mlvec{M}(\mlvec\theta)}{dt}
      =& \sum_{j=1}^{m} r_{j} [\mlvec{M}(\mlvec\theta), [\mlvec{M}(\mlvec\theta), \mlvec{\rho}_{j}]]\\
      +& \sum_{j=1}^{m}r_{j} [\mlvec{M}(\mlvec\theta), \tr_{1}\big(([\mlvec{M}(\mlvec\theta), \mlvec{\rho}_{j}]\otimes\mlvec{I}) (\mlvec{Y}(\mlvec\theta(t)) - \mlvec{Y}^{\star})\big)]
  \end{align}
  Trace after multiplying $\mlvec{\rho}_{i}$ on both sides:
  \begin{align}
    \frac{dr_{i}}{dt} = -\tr(\mlvec{\rho}_{i}\frac{d\mlvec{M}(\mlvec\theta)}{dt}) =
      & - \sum_{j=1}^{m} r_{j} \tr\big(\mlvec{\rho}_{i} [\mlvec{M}(\mlvec\theta), [\mlvec{M}(\mlvec\theta), \mlvec{\rho}_{j}]]\big)\\
      & - \sum_{j=1}^{m}r_{j} \tr\big(\mlvec{\rho}_{i} [\mlvec{M}(\mlvec\theta), \tr_{1}\big(([\mlvec{M}(\mlvec\theta), \mlvec{\rho}_{j}]\otimes\mlvec{I}) (\mlvec{Y}(\mlvec\theta(t)) - \mlvec{Y}^{\star})\big)]\big)
  \end{align}
  The lemma follows directly from rearranging: for the first term,
  \begin{align}
    &-\sum_{j=1}^{m}r_{j}\tr(\mlvec{\rho}_{i}[\mlvec{M}(\mlvec\theta), [\mlvec{M}(\mlvec\theta), \mlvec{\rho}_{j}]])\\
    =&-\sum_{j=1}^{m}r_{j}\tr([\mlvec{\rho}_{i}, \mlvec{M}(\mlvec\theta)][\mlvec{M}(\mlvec\theta), \mlvec{\rho}_{j}])\\
    =&-\sum_{j=1}^{m}r_{j}\tr(\imagi[\mlvec{M}(\mlvec\theta),\mlvec{\rho}_{i}]\imagi[\mlvec{M}(\mlvec\theta), \mlvec{\rho}_{j}]).
  \end{align}
  For the second term,
  \begin{align}
    & - \sum_{j=1}^{m}r_{j} \tr\big(\mlvec{\rho}_{i} [\mlvec{M}(\mlvec\theta), \tr_{1}\big(([\mlvec{M}(\mlvec\theta), \mlvec{\rho}_{j}]\otimes\mlvec{I}) (\mlvec{Y}(\mlvec\theta(t)) - \mlvec{Y}^{\star})\big)]\big)\\
    =& - \sum_{j=1}^{m}r_{j} \tr\big( \imagi[\mlvec{M}(\mlvec\theta), \mlvec{\rho}_{i}] \tr_{1}\big((\imagi[\mlvec{M}(\mlvec\theta), \mlvec{\rho}_{j}]\otimes\mlvec{I}) (\mlvec{Y}(\mlvec\theta(t)) - \mlvec{Y}^{\star})\big)\big)\\
    =& - \sum_{j=1}^{m}r_{j} \tr\big((\mlvec{I}\otimes\imagi[\mlvec{M}(\mlvec\theta), \mlvec{\rho}_{i}]) (\imagi[\mlvec{M}(\mlvec\theta), \mlvec{\rho}_{j}]\otimes\mlvec{I}) (\mlvec{Y}(\mlvec\theta(t)) - \mlvec{Y}^{\star})\big)\\
    =& - \sum_{j=1}^{m}r_{j} \tr\big((\imagi[\mlvec{M}(\mlvec\theta), \mlvec{\rho}_{j}]\otimes\imagi[\mlvec{M}(\mlvec\theta), \mlvec{\rho}_{i}]) (\mlvec{Y}(\mlvec\theta(t)) - \mlvec{Y}^{\star})\big)\\
    =& - \sum_{j=1}^{m}r_{j} \tr\big((\imagi[\mlvec{M}(\mlvec\theta), \mlvec{\rho}_{i}]\otimes\imagi[\mlvec{M}(\mlvec\theta), \mlvec{\rho}_{j}]) (\mlvec{Y}(\mlvec\theta(t)) - \mlvec{Y}^{\star})\big)
  \end{align}
  The last equality follows from the fact that $\mlvec{Y}$ and
  $\mlvec{Y}^{\star}$ are invariant under the swapping of spaces.
  The lemma follows by identifying the matrix $\Delta(t)$ with $\mlvec{Y}(\mlvec\theta(t))-\mlvec{Y}^\star$.
\end{proof}

\subsection{Proof of Theorem~\ref{thm:qnn-mse-linear}}
\label{subsec:thm_qnn_mse_linear}
\qnnmselinear*
\begin{proof}
In Lemma~\ref{lm:resid_dyn_decomp}, we decompose the QNN dynamics into the
asymptotic term and the perturbation term depending on
$\Delta(t) = \mlvec{Y}(\mlvec\theta(t)) - \mlvec{Y}^{\star}$. We now show that the use of
the terms
``asymptotic'' and ``perturbation'' are exact, by showing that
$\mlvec{Y}(\mlvec\theta(t)) - \mlvec{Y}^{\star}$ vanishes as
$p\rightarrow \infty$. We make use of the characterization of a similarly-defined quantity in \cite{ourvqe2022}, restated as
Lemma~\ref{lm:concentration_init} and \ref{lm:concentration_train}, such that
for sufficiently large $p$,
$\opnorm{\mlvec{Y}(\mlvec\theta(t)) - \mlvec{Y}^{\star}}$ vanishes for all $t$
with high probability over the randomness in $\{\mlvec{U}_{l}\}_{l=0}^{p}$.
Recall that the perturbation term $\Kpert$ is defined as
\begin{align}
  (\Kpert(t))_{ij} &:=
      \tr\big((\imagi[\mlvec{M}(\mlvec\theta), \mlvec{\rho}_{i}]\otimes\imagi[\mlvec{M}(\mlvec\theta), \mlvec{\rho}_{j}]) (\mlvec{Y}(\mlvec\theta(t)) - \mlvec{Y}^{\star})\big).
\end{align}
By choosing sufficiently large $p$, we have $\opnorm{\Kpert(t)} \leq C_{0} / 10$
and therefore the loss function converging to zero at a rate $\geq C_{0} / 2$.
\end{proof}
\begin{lemma}[Concentration at initialization, adapted from Lemma 3.4 in \cite{ourvqe2022}]
  \label{lm:concentration_init}
     Over the randomness of ansatz initialization (i.e. for
     $\{\mlvec{U}_{l}\}_{l=1}^{p}$ sampled $i.i.d.$ with respect to
     the Haar measure), for any initial $\mlvec\theta(0)$, with probability $1 - \delta$:
     \begin{align}
     \opnorm{\YY(\mlvec\theta(0)) - \YY^{\star}} \leq \frac{1}{\sqrt{p}}\cdot \frac{2\opnorm{\mlvec{H}}^{2}}{Z}\sqrt{2\log\frac{d^{2}}{\delta}}.
     \end{align}
\end{lemma}
\begin{proof}
Define
\begin{align}
\mlvec{X}_{l}:=\frac{1}{Z(\mlvec{H},d)}\big(\mlvec{U}_{0:l-1}(\mlvec\theta(0))^{\dagger}\mlvec{H}\mlvec{U}^{\dagger}_{0:l-1}(\mlvec\theta(0))\big)^{\otimes 2} - \mlvec{Y}^{\star}.
\end{align}
By straight-forward calculation (e.g. using results in \citet{collins2006integration}) we know that $X_{l}$ is centered (i.e
$\EXP[X_{l}] = 0$). The set $\{\mlvec{X}_{l}\}$ can be viewed as independent random matrices as the Haar
random unitary removes all the correlation. The matrix on the left-hand side can
therefore be expressed as the arithmetic average of $p$ independent random matrices.
The square of $\mlvec{X}_{l}$ is bounded in operator norm:
\begin{align}
  \opnorm{\mlvec{X}_{l}^{2}} = \opnorm{\mlvec{X}_{l}}^{2}\leq (\frac{\opnorm{\mlvec{H}}^{2}}{Z} + \frac{d+1}{d})^{2} \leq (\frac{2\opnorm{\mlvec{H}}^{2}}{Z(\mlvec{H},d)})^{2}
 \end{align}
where the second inequality follows from the fact that the ratio $g_{1} = \opnorm{\HH}^{2} / \tr(\HH^{2})$ satisfies that $1 \geq g_{1} \geq 1/d$.
By Hoeffding's inequality(\cite{tropp2012user}, Thm 1.3), with probability
$\geq 1-\delta$,
\begin{align}
  \opnorm{\YY(\mlvec\theta(0)) - \YY^{\star}} \leq \frac{1}{\sqrt{p}}\cdot \frac{2\opnorm{\mlvec{H}}^{2}}{Z(\mlvec{H},d)}\sqrt{\log\frac{2d^{2}}{\delta}}.
\end{align}
\end{proof}

As we pointed out in the main body, a vanishing perturbation term at initialization is not
sufficient to guarantee the term remain perturbative throughout the training. We
now show in Lemma~\ref{lm:concentration_train} that
$\mlvec{Y}(\mlvec\theta(t)) - \mlvec{Y}^{\star}$ remain small during training by
showing $\mlvec{Y}(\mlvec\theta(t)) - \mlvec{Y}(\mlvec(0))$ vanishes in the
limit $p \rightarrow \infty$. But before that, we show that, while the QNN
predictions changes much during training, the change in the parameters measured
in $\ell_{2}$- or $\ell_{\infty}$-norm
($\|\mlvec{\theta}(t) - \mlvec{\theta}(0)\|_{2}$ or
$\|\mlvec{\theta}(t) - \mlvec{\theta}(0)\|_{\infty}$) vanishes as
$p \rightarrow \infty$ during the training of QNN:
\begin{restatable}[Slow-varying $\theta$ in QNNs]{lemma}{qnnslowtheta}
  \label{lm:qnn-slow-theta}
  Suppose that under learning rate $\eta=\frac{m}{p Z(\mlvec{H},d)}$, for all $0 \le t \le T$,
  the loss function $L(\mlvec\theta(t)) \leq  L(\mlvec\theta(0)) \exp(-at)$ for
  some constant $a$,
  then for all $0 \le t_1,t_2 \le T$:
  \begin{align}
    \|\mlvec{\theta}(t_2) - \mlvec{\theta}(t_1)\|_{\infty}
    &\leq \frac{1}{p}\frac{\sqrt{2}m\fronorm{\mlvec{H}}\fronorm{\mlvec{M}}\sqrt{L(\mlvec\theta(0))}}{Z}|t_{1} - t_{2}|,\\
    \|\mlvec{\theta}(t_2) - \mlvec{\theta}(t_1)\|_{2}
    &\leq \frac{1}{\sqrt{p}}\frac{\sqrt{2}m\fronorm{\mlvec{H}}\fronorm{\mlvec{M}}\sqrt{L(\mlvec\theta(0))}}{Z}|t_{1} - t_{2}|.
  \end{align}
\end{restatable}
\begin{proof}
  We first bound the absolute value of the derivative $\frac{d\theta_{l}}{dt}$:
  \begin{align}
    |\frac{d\theta_{l}}{dt}| = \eta |\frac{\partial L}{\partial \theta_{l}}|
    = \frac{\eta}{2m}|\sum_{i=1}^{m}r_{i}\tr(\imagi[\mlvec{M}(\mlvec\theta), \mlvec{H}_{l}]\mlvec\rho_{i})|.
  \end{align}
  Plugging in $\eta = \frac{m}{p Z}$, we have
  \begin{align}
    |\frac{d\theta_{l}}{dt}| = \frac{1}{2pZ}|\sum_{i=1}^{m}r_{i}\tr(\imagi[\mlvec{M}(\mlvec\theta), \mlvec{H}_{l}]\mlvec\rho_{i})|
    = \frac{1}{2pZ}|\<\mlvec{r}, \mlvec{a}\>|,
  \end{align}
  where the vector $\mlvec{a}$ is defined such that
  $a_{j} = \tr(i[\mlvec{M}(\mlvec\theta), \mlvec{H}_{l}]\mlvec{\rho}_{j})$ for
  $j \in [m]$. The $\ell_{2}$-norm of $\mlvec{a}$
  \begin{align}
    \|\mlvec{a}\|_{2}^{2} &= \sum_{j=1}^{m}\tr^{2}(\imagi [\mlvec{M}, \mlvec{H}_{l}]\mlvec{\rho}_{j})\\
    &= \tr\big((i[\mlvec{M}, \mlvec{H}_{l}])^{\otimes 2} \sum_{j=1}^{m}\mlvec{\rho_{j}}^{\otimes 2}\big)\\
    &\leq \fronorm{(i[\mlvec{M}, \mlvec{H}_{l}])^{\otimes 2}}\fronorm{\sum_{j=1}^{m}\mlvec{\rho_{j}}^{\otimes 2}}\\
    &\leq \fronorm{i[\mlvec{M}, \mlvec{H}_{l}]}^{2}\fronorm{\sum_{j=1}^{m}\mlvec{\rho_{j}}^{\otimes 2}}\\
    &\leq (2\fronorm{\mlvec{M}}\fronorm{\mlvec{H}_{l}})^{2}\sum_{j=1}^{m}\fronorm{\mlvec{\rho_{j}}^{\otimes 2}}\\
    &\leq (2 \fronorm{\mlvec{M}}\fronorm{\mlvec{H}} \sqrt{m})^{2}.
  \end{align}
  Therefore we can bound
  $|\frac{d\theta_{l}}{dt}|$ as
  \begin{align}
    |\frac{d\theta_{l}}{dt}|
    &\leq \frac{1}{2pZ}\|\mlvec{r}\|_{2}\|\mlvec{a}\|_{2}\\
    &\leq \frac{1}{2pZ}\sqrt{2mL(\mlvec\theta(t))} \cdot 2 \fronorm{\mlvec{M}}\fronorm{\mlvec{H}} \sqrt{m}\\
    &= \frac{1}{p}\frac{\sqrt{2} m \fronorm{\mlvec{M}}\fronorm{\mlvec{H}}}{Z} \sqrt{L(\mlvec\theta(t))}\\
    &\leq \frac{1}{p}\frac{\sqrt{2} m \fronorm{\mlvec{M}}\fronorm{\mlvec{H}}}{Z} \sqrt{L(\mlvec\theta(0))} \exp(-at / 2)
  \end{align}
  Hence for all $l\in[p]$:
  \begin{align}
    |\theta_{l}(t_{2}) - \theta_{l}(t_{1})|
    &= |\int_{t_{1}}^{t_{2}}dt d\theta_l(t)/dt| \leq \int_{t_{1}}^{t_{2}}dt |d\theta_l(t)/dt|\\
    &\leq \int_{t_{1}}^{t_{2}}dt \frac{1}{p}\frac{\sqrt{2} m \fronorm{\mlvec{M}}\fronorm{\mlvec{H}}}{Z} \sqrt{L(\mlvec\theta(0))} \exp(-at / 2)\\
    &\leq \frac{2}{a} \cdot \frac{1}{p}\frac{\sqrt{2} m \fronorm{\mlvec{M}}\fronorm{\mlvec{H}}}{Z} \sqrt{L(\mlvec\theta(0))} |\exp(-at_{1} / 2) - \exp(-at_{2}/2)|\\
    &\leq \frac{1}{p}\frac{\sqrt{2} m \fronorm{\mlvec{M}}\fronorm{\mlvec{H}}}{Z} \sqrt{L(\mlvec\theta(0))} |t_{1}-t_{2}|\\
  \end{align}
  The bounds on the $\ell_{2}$- and $\ell_{\infty}$-norm follows from direct computation.
\end{proof}

We are now ready to show $\mlvec{Y}(t_{2}) - \mlvec{Y}(t_{1})$ vanishes as
$p\rightarrow \infty$:
\begin{lemma}[Concentration during training, adapted from Lemma 3.5 in \cite{ourvqe2022}]
  \label{lm:concentration_train}
      Suppose that under learning rate $\eta=\frac{m}{p Z(\mlvec{H},d)}$, for
      all $0 \le t \le T$, the loss function $L(\mlvec\theta(t))$ decreases as $L(\mlvec\theta(0))\exp(-at)$
       then with probability $\geq 1-\delta$, for all $0 \leq t\leq T$:
      $\opnorm{\mlvec{Y}(\mlvec{\theta}(t)) - \mlvec{Y}(\mlvec{\theta}(0))} \leq C_3\cdot\frac{T}{\sqrt{p}}$,
      where $C_3$ is a constant of $T$ and $p$.
\end{lemma}
\begin{proof}
  To bound the supremum of the matrix-valued random field, we use an adapted
  version of the Dudley's inequality:

  \textit{Claim 1 (Dudley's inequality for matrix-valued random fields, adapted
    from Theorem 8.1.6 in High-dimensional probability (Vershynin, 2018).).} Let $\mlvec{\mathcal{R}}$ be a metric space equipped with a metric $\mathbf{d}(\cdot,\cdot)$, and $\mlvec{X}: \mlvec{\mathcal{R}} \mapsto \real^{D \times D}$ with subgaussian increments i.e. it satisfies
    $\Prob[\| \mlvec{X}(r_1) - \mlvec{X}(r_2)\|_{\mathsf{op}} > t] \le 2D\exp\left(-\frac{t^2}{C_{\sigma}^2\mathbf{d}(r_1,r_2)^2}\right)$.
  Then with probability at least $1 - 2D\exp(-u^2)$ for any subset $\mathcal{S} \subseteq \mlvec{\mathcal{R}}$:
   $
    \sup_{(r_1,r_2) \in \mathcal{S}} \|\mlvec{X}(r_1) - \mlvec{X}(r_2)\|_{\mathsf{op}} \le C \cdot C_{\sigma} \left[\int_0^{\mathrm{diam}(\mathcal{S})} \sqrt{\mathcal{N}(\mathcal{S},\mathbf{d},\epsilon)}\,d\epsilon + u \cdot \mathrm{diam}(\mathcal{S})\right]
   $
  for some constant $C$, where $\mathcal{N}(\mathcal{S},\mathbf{d},\epsilon)$ is the metric entropy defined as the logarithm of the $\epsilon$-covering number of $\mathcal{S}$ using metric $d$.

  To make use of \emph{Claim 1}, we now establish the sub-gaussian increment of
  $\mlvec{Y}(\mlvec\theta(t))$ through the following \emph{Claim 2} by applying \emph{McDiarmid inequality}:

  \textit{Claim 2 (Sub-gaussianity of $\mlvec{Y}$)}
    $\Prob[\opnorm{\mlvec{Y}(\mlvec{\theta}) - \mlvec{Y}(\mlvec{0})} > t] \le 2\exp\left(-\frac{-t^2 Z(\mlvec{H},d)^2}{2C_1\|\mlvec{\theta}\|_{2}^2}\right)$
    for some constant $C_{1}$. Then due to the Haar distribution of the unitaries $\{\mlvec{U}_{l}\}_{l=0}^{p}$,
    \begin{align}
    \Prob[\opnorm{\mlvec{Y}(\mlvec{\theta}_2) - \mlvec{Y}(\mlvec{\theta}_1)} > t] \le 2\exp\left(-\frac{-t^2 Z(\mlvec{H},d)^2}{2C_1\|\mlvec{\theta}_2 - \mlvec{\theta}_1\|_{2}^2}\right).
    \end{align}

  To see that \emph{Claim 2} is true, consider an alternative description of
  $\mlvec{Y}(\mlvec\theta)$. Recall that $\mlvec{Y}(\mlvec\theta)$ is defined as
  $\mlvec{Y}(\mlvec{\theta}) = \frac{1}{pZ(\mlvec{H},d)}\sum_{l=1}^{p} \mlvec{Y}_l$
  with$\mlvec{Y}_l(\mlvec{\theta})$ being $\mlvec{H}_l^{\otimes 2}$.
  We consider a re-parameterization of the random variables
  $\mlvec{H}_l(\theta)$ by constructing random variables that are identically
  distributed, but are functions on a different latent probability space.
  Defining $\mlvec{H}_{l}$ as
  $\mlvec{U}_{0}^{\dagger}\cdots \mlvec{U}_{l-1}^{\dagger}\mlvec{H}\mlvec{U}_{l-1}\cdots\mlvec{U}_{0}$,  $\mlvec{Y}$ can be rewritten as:
  \begin{align}
    \mlvec{Y}(\mlvec\theta) = \frac{1}{pZ} \sum_{l=1}^{p}
    \big(
    e^{i\theta_{1}\mlvec{H}_{1}}\cdots e^{i\theta_{l-1}\mlvec{H}_{l-1}}
    \mlvec{H}_{l}
    e^{-i\theta_{l-1}\mlvec{H}_{l-1}}\cdots e^{-i\theta_{1}\mlvec{H}_{1}}
    \big)^{\otimes 2}.
  \end{align}
  By the Haar randomness of $\{\mlvec{U}_{l}\}_{l=1}^{p}$, we can view
  $\{\mlvec{H}_{l}\}_{l=1}^{p}$ as random Hermitians generated by
  $\{\mlvec{V}_{l}\mlvec{H}\mlvec{V}^{\dagger}_{l}\}$ for \textit{i.i.d.} Haar
  random $\{\mlvec{V}_{l}\}_{l=1}^{p}$. This variable is identically distributed to $\mlvec{Y}$ and $\mlvec{Y}_{l}$ can be defined as each term in the sum.

We will apply the well-known McDiarmid inequality (e.g. Theorem~2.9.1 in High-dimensional probability (Vershynin, 2018)) that can be stated as follows: Consider independent random variables $X_1,\dots,X_k \in \mathcal{X}$. Suppose a random variable $\phi \colon \mathcal{X}^{k} \to \real$ satisfies the condition that for all $1 \le j \le k$ and for all $x_1,\dots,x_j,\dots,x_k,x'_j \in \mathcal{X}$,
\begin{align}
    |\phi(x_1,\dots,x_j,\dots,x_k) - \phi(x_1,\dots,x'_j,\dots,x_k)| \le c_j,
\end{align}
then the tails of the distribution satisfy
\begin{align}
    \Prob[|\phi(X_1,\dots,X_k) - \EXP\phi| \ge t] \le \exp\left(\frac{-2t^2}{\sum_{i=1}^{k}c_i^2}\right).
\end{align}

With our earlier re-parameterization we can consider $\mlvec{Y}$ and consequently $\mlvec{Y}_l$ as functions of the randomly sampled Hermitian operators $\mlvec{H}_l$. Define the variable $\mlvec{Y}^{(k)}$ as that obtained by resampling $\mlvec{H}_k$ independently, and $\mlvec{Y}_l^{(k)}$ correspondingly. Finally we define
\begin{align}
\Delta^{(k)} \mlvec{Y} = \opnorm{(\mlvec{Y}(\mlvec{\theta}) - \mlvec{Y}(0)) - (\mlvec{Y}^{(k)}(\mlvec{\theta}) - \mlvec{Y}^{(k)}(0))} = \opnorm{\mlvec{Y}(\mlvec{\theta}) - \mlvec{Y}^{(k)}(\mlvec{\theta})}.
\end{align}
Via the triangle inequality,
\begin{align}
    \Delta^{(k)} \mlvec{Y} &= \lVert \mlvec{Y}(\mlvec{\theta}) - \mlvec{Y}^{(k)}(\mlvec{\theta}) \rVert
    = \frac{1}{pZ}\lVert \sum_{l\geq k} \mlvec{Y}_{l}(\mlvec{\theta}) - \mlvec{Y}_{l}^{(k)}(\mlvec{\theta}) \rVert \\
    &\le \frac{1}{pZ} \sum_{l\geq k} \lVert \mlvec{Y}_{l}(\mlvec{\theta}) - \mlvec{Y}_{l}^{(k)}(\mlvec{\theta}) \rVert.
\end{align}
Then by definition,
    \begin{align*}
    &\|\mathbf{Y}_{l}(\mathbf{\theta}) - \mathbf{Y}_{l}^{(k)}(\boldsymbol{\theta})\| \nonumber\\
    =& \|(e^{\imagi\theta_1\mathbf{H}_1}\cdots e^{\imagi\theta_{k-1}\mathbf{H}_{k-1}})^{\otimes 2}
    \big((e^{\imagi{\theta}_{k}\mathbf{H}_{k}} \mathbf{K} e^{-\imagi{\theta}_{k}\mathbf{H}_{k}})^{\otimes 2}\nonumber\\
    -&(e^{\imagi{\theta}_{k}\mathbf{H}^{\prime}_{k}} \mathbf{K} e^{-\imagi{\theta}_{k}\mathbf{H}^{\prime}_{k}})^{\otimes 2}
     \big)
     (e^{-\imagi{\theta}_{k-1}\mathbf{H}_{k-1}}\cdots e^{-\imagi{\theta}_{1}\mathbf{H}_{1}})^{\otimes 2}\|\\
     =&\|(e^{\imagi{\theta}_{k}\mathbf{H}_{k}} \mathbf{K} e^{-\imagi{\theta}_{k}\mathbf{H}_{k}})^{\otimes 2} -(e^{\imagi{\theta}_{k}\mathbf{H}^{\prime}_{k}} \mathbf{K} e^{-\imagi{\theta}_{k}\mathbf{H}^{\prime}_{k}})^{\otimes 2} \|\\
     \leq&\|(e^{\imagi{\theta}_{k}\mathbf{H}_{k}} \mathbf{K} e^{-\imagi{\theta}_{k}\mathbf{H}_{k}})^{\otimes 2} - \mathbf{K}^{\otimes 2}\|+\|(e^{\imagi{\theta}_{k}\mathbf{H}^{\prime}_{k}} \mathbf{K} e^{-i{\theta}_{k}\mathbf{H}^{\prime}_{k}})^{\otimes 2} -\mathbf{K}^{\otimes2}\|.
\end{align*}
where $\mathbf{K}:=
    e^{\imagi{\theta}_{k+1}\mathbf{H}_{k+1}}\cdots e^{\imagi{\theta}_{l-1}\mathbf{H}_{l-1}}
    \mathbf{H}_{l}
    e^{-\imagi{\theta}_{l-1}\mathbf{H}_{l-1}}\cdots e^{-\imagi{\theta}_{k+1}\mathbf{H}_{k+1}}$.
    Let $\mathbf{K}(\phi)$ denote $e^{\imagi\phi\mathbf{H}_k}\mathbf{K} e^{-\imagi\phi\mathbf{H}_k}$, we can bound the first term on the righthand side as follows:
    \begin{align*}
        &\|(e^{\imagi{\theta}_{k}\mathbf{H}_{k}} \mathbf{K} e^{-\imagi{\theta}_{k}\mathbf{H}_{k}})^{\otimes 2} - \mathbf{K}^{\otimes 2}\|\\
        =& \| \mathbf{K}(\theta_k)^{\otimes 2} - \mathbf{K}(0)^{\otimes 2}\|\\
        =&\|\int_0^{\theta_k}d\phi \frac{d}{d\phi}(\mathbf{K}(\phi)^{\otimes 2})\|\\
        \leq&\int_0^{\theta_k}d\phi\| \frac{d}{d\phi}(\mathbf{K}(\phi)^{\otimes 2})\|\\
        \leq & 4|\theta_k|\|\mathbf{H}_k\|\|\mathbf{K}\|^2.
    \end{align*}
    The last inequality follows from the fact that
    \begin{align*}
        &\|\frac{d}{d\phi} \mathbf{K}(\phi)^{\otimes 2}\| \\
        =& \|(\exp(\imagi\phi\mathbf{H}_k))^{\otimes 2}
        \big(
        [-\imagi\mathbf{H}_k, \mathbf{K}]\otimes \mathbf{K}
        +
        \mathbf{K}\otimes[-\imagi\mathbf{H}_k, \mathbf{K}]
        \big)
        (\exp(-\imagi\phi\mathbf{H}_k))^{\otimes 2}\|\\
        =& \|
        [-\imagi\mathbf{H}_k, \mathbf{K}]\otimes \mathbf{K}
        +
        \mathbf{K}\otimes[-\imagi\mathbf{H}_k, \mathbf{K}]\|\\
        \leq& 4\|\mathbf{H}_k\|\|\mathbf{K}\|^2.
    \end{align*}
    The same reasoning holds for the term with $\mathbf{H}'_k$. Using the fact
    that $\|\mathbf{H}_k\|=\|\mathbf{H}'_k\| = \|\mathbf{H}\|$, and we have
    \begin{align*}
    \|\big(\mathbf{Y}_{l}(\boldsymbol{\theta}) - \mathbf{Y}_{l}(\boldsymbol{0})\big)
  - \big(\mathbf{Y}^{(k)}_{l}(\boldsymbol{\theta}) - \mathbf{Y}^{(k)}_{l}(\boldsymbol{0})\big)\|
  \leq 8|{\theta}_{k}|\|\mathbf{H}\|\|{\mathbf{K}\|^{2}
  = 8|\theta}_{k}|\|\mathbf{H}\|^{3}.
    \end{align*}
\emph{Claim 2} follows from the direct application of
McDiarmid inequality.

By Lemma~\ref{lm:qnn-slow-theta},
$\|\mlvec{\theta}(t_2) - \mlvec{\theta}(t_1)\|_2 \le \frac{C_{L}}{\sqrt{p}}|t_{2}-t_{1}|$
with $C_{L}$ being a constant with respect to $p$. Plugging this into
\textit{Claim 2}, we see that $\mlvec{Y}$ has sub-gaussian increments if we
define the metric
$\mathbf{d}(t_2,t_1) = \frac{C_{L}}{\sqrt{p}}\cdot |t_2 - t_1|$, thereby
satisfying the conditions for \textit{Claim 1}. Under this metric, the diameter
of the interval $[0,T]$ is of order $\frac{T}{\sqrt{p}}$. Applying \emph{Claim 1},
with $u = \sqrt{\log(2d/\delta)}$ to ensure a failure probability at most
$\delta$ we have
\begin{align}
  \sup_{t \in [0,T]} \|\mlvec{Y}(\theta(t)) - \mlvec{Y}(\theta(0))\|_{\mathsf{op}} \le C_3\cdot\frac{T}{\sqrt{p}},
\end{align}
where $C_3$ is a constant of $p$ and $T$ and depends polynomially on other
quantities including $d$ and $\log(1/\delta)$.
\end{proof}


\section{Proof for Theorem~\ref{thm:smallest_eig}}
\label{sec:app_globalminima}
In this section, we present the proof for Theorem~\ref{thm:smallest_eig} for
characterizing the rate of convergence at global minima:
\asympeig*

We start by presenting a few helper lemma:
\subsection{Helper lemma for \texorpdfstring{$\Kasym$}{Kasym}}
\begin{lemma}
  \label{lm:hadamard}
  Let $\mlvec{A},\mlvec{B}$ be $d\times d$ Hermitians. Let $\opnorm{\cdot}$
  denote the operator norm of a given Hermitian and let $\circ$ denote the
  Hadamard product (i.e. the elementwise multiplication) of two matrices, we have
  \begin{align}
    \opnorm{\mlvec{A}\circ\mlvec{B}} \leq \opnorm{\mlvec{A}}\opnorm{\mlvec{B}}.
  \end{align}
\end{lemma}
\begin{proof}
  For any $d\times d$ Hermitian matrix, let $\lambda_{i}(\cdot)$ denote its $i$-th
  smallest eigenvalue. The Hadamard product $\mlvec{A}\circ\mlvec{B}$ is a
  $d\times d$ principal submatrix of the Kronecker product
  $\mlvec{A}\otimes \mlvec{B}$, and by the Poincar\'e separation theorem (see
  e.g. Corollary 4.3.37 in \citet{horn2012matrix}):
  \begin{align}
    \lambda_{1}(\mlvec{A}\otimes \mlvec{B})\leq \lambda_{i}(\mlvec{A}\circ\mlvec{B})\leq \lambda_{d^{2}}(\mlvec{A}\otimes \mlvec{B}).
  \end{align}
  The statement follows from the fact that the eigenvalues of
  $\mlvec{A}\otimes \mlvec{B}$ take the form of
  $\lambda_{i}(\mlvec{A})\lambda_{j}(\mlvec{B})$ for $i,j\in[d]$.
\end{proof}
\begin{restatable}[$\Kasym$ for asymptotic dynamics]{lemma}{asympestimation}
\label{lm:asympkernelestimate}
Let $\mathcal{S}$ be a $m$-sample training set composed of pure states
$\{\mlvec{\rho}_{j} = \mlvec{v}_{j}\mlvec{v}_{j}^{\dagger}\}_{j=1}^{m}$. Let
$\qnnmeasure$ be a Pauli-like measurement with eigenvalues $\pm 1$ and
trace-$0$. Consider training a QNN with $\mathcal{S}$, measurement $\qnnmeasure$
and a scaling factor of $\gamma$.
The positive semidefinite matrix $\Kasym$ can be expressed entry-wise as
\begin{align}
  (\Kasym)_{ij}(\mlvec{M}(t)) = 8 \gamma^{2}Re(\mlvec{u}_{j}^{\dagger}(t)\mlvec{u}_{i}(t) \mlvec{w}_{i}^{\dagger}(t)\mlvec{w}_{j}(t)),
\end{align}
where $\mlvec{u}_{i}(t):= \mlvec{\Pi}_{+}(t)\mlvec{v}_{i} $ (resp.
$\mlvec{w}_{i}(t):=\mlvec{\Pi}_{-}(t)\mlvec{v}_{i}$) is the projection of
$\mlvec{v}_{i}$ into the postive (resp. negative) subspace of
$\mlvec{M}(t)=\gamma(\mlvec{\Pi}_{+}(t) - \mlvec{\Pi}_{-}(t))$.
Let $\mlvec{P}(t):=(\mlvec{u}_{i}^{\dagger}(t)\mlvec{u}_{j}(t))_{i.j\in[m]}$ and
$\mlvec{N}(t):=(\mlvec{w}_{i}^{\dagger}(t)\mlvec{w}_{j}(t))_{i,j\in[m]}$ be the Gram matrices of
$\{\mlvec{u}_{i}(t)\}_{i=1}^{m}$ and $\{\mlvec{w}_{i}(t)\}_{i=1}^{m}$, we have:
\begin{align}
  \lambda_{\min}(\Kasym(t)) \geq  8\gamma^{2}\lambda_{\min}(\mlvec{P}(t)) \min_{i\in[m]}(\mlvec{N}_{ii}(t)) \geq 8\gamma^{2}\lambda_{\min}(\mlvec{P}(t)) \lambda_{\min}(\mlvec{N}(t)).
\end{align}
\end{restatable}
\begin{proof}
  For succinctness, we drop the time dependency $t$ when there are no ambiguities.
Calculate the expression of $(\Kasym)_{ij}$ for pure states
$\mlvec{\rho}_{i} = \mlvec{v}_{i}\mlvec{v}_{i}^{\dagger}$:
\begin{align}
  (\Kasym(\mlvec{M}(t)))_{ij}
  &= \tr\big(\imagi[\mlvec{M}, \mlvec{\rho}_{i}] \ \imagi[\mlvec{M}, \mlvec{\rho}_{j}]\big)\\
  &= \tr\big(\mlvec{M}^{2}\mlvec{\rho}_{i} \mlvec{\rho}_{j}\big) + \tr\big(\mlvec{M}^{2}\mlvec{\rho}_{j} \mlvec{\rho}_{i}\big)
    - 2\tr\big(\mlvec{M}\mlvec{\rho}_{i}\mlvec{M} \mlvec{\rho}_{j}\big)\\
  &= 2\gamma^{2}\big(\tr(\mlvec{\rho}_{i} \mlvec{\rho}_{j}\big)
    -
    \tr((\mlvec{\Pi}_{+} - \mlvec{\Pi}_{-})\mlvec{\rho}_{i}(\mlvec{\Pi}_{+} - \mlvec{\Pi}_{-}) \mlvec{\rho}_{j})
    \big)
\end{align}
Plugging in $\mlvec{\rho}_{i} = \mlvec{v}_{i}\mlvec{v}_{i}^{\dagger}$, we have:
\begin{align}
  \frac{1}{2\gamma^{2}}(\Kasym(\mlvec{M}(t)))_{ij}
  &= |\mlvec{u}_{i}^{\dagger}\mlvec{u}_{j} + \mlvec{w}_{i}^{\dagger}\mlvec{w}_{j}|^{2} - |(\mlvec{u}_{i} + \mlvec{w}_{i})^{\dagger}(\mlvec{\Pi}_{+} - \mlvec{\Pi}_{-}) (\mlvec{u}_{j} + \mlvec{w}_{j})|^{2}\\
  &= |\mlvec{u}_{i}^{\dagger}\mlvec{u}_{j} + \mlvec{w}_{i}^{\dagger}\mlvec{w}_{j}|^{2} - |(\mlvec{u}_{i} + \mlvec{w}_{i})^{\dagger} (\mlvec{u}_{j} - \mlvec{w}_{j})|^{2}\\
  &= |\mlvec{u}_{i}^{\dagger}\mlvec{u}_{j} + \mlvec{w}_{i}^{\dagger}\mlvec{w}_{j}|^{2} - |\mlvec{u}_{i}^{\dagger}\mlvec{u}_{j} - \mlvec{w}_{i}^{\dagger}\mlvec{w}_{j}|^{2}\\
  &= 2 \mlvec{u}_{i}^{\dagger}\mlvec{u}_{j} \cdot \mlvec{w}_{j}^{\dagger}\mlvec{w}_{i} + 2 \mlvec{u}_{j}^{\dagger}\mlvec{u}_{i} \cdot \mlvec{w}_{i}^{\dagger}\mlvec{w}_{j}\\
  &= 4 Re(\mlvec{u}_{j}^{\dagger}\mlvec{u}_{i} \mlvec{w}_{i}^{\dagger}\mlvec{w}_{j}),
\end{align}
or $(\Kasym(\mlvec{M}(t)))_{ij} = {8\gamma^{2}} Re(\mlvec{u}_{j}^{\dagger}\mlvec{u}_{i} \mlvec{w}_{i}^{\dagger}\mlvec{w}_{j})$.

Let $\mlvec{P}(t)$ and $\mlvec{N}(t)$ be the Gram matrices for $\{\mlvec{u}_{i}(t)\}_{i=1}^{m}$
and $\{\mlvec{w}_{i}(t)\}_{i=1}^{m}$:
\begin{align}
  (\mlvec{P}(t))_{ij}  = \mlvec{u}(t)_{i}^{\dagger}\mlvec{u}(t)_{j},\ (\mlvec{N}(t))_{ij}  = \mlvec{w}(t)_{i}^{\dagger}\mlvec{w}(t)_{j},
\end{align}
the matrix $\Kasym$ can be expressed as
$\Kasym = 4\gamma^{2} \mlvec{P}\circ\mlvec{N}^{T} + 4\gamma^{2} \mlvec{P}^{T}\circ\mlvec{N}$,
where $\circ$ denotes the Hadamard product, with $\mlvec{P}$ and $\mlvec{N}$
being positive semidefinite matrices. Following a result of Schur's (e.g. see
Lemma 6.5 in \cite{oymak2020toward}), we estimate the smallest eigenvalue of
$\Kasym$ as
\begin{align}
  \lambda_{\min}(\Kasym(\mlvec{M}(\mlvec\theta))) \geq 8\gamma^{2}\max\big(\min_{i\in[m]}(\mlvec{N}_{ii})\lambda_{\min}(\mlvec{P}), \min_{i\in[m]}(\mlvec{P}_{ii})\lambda_{\min}(\mlvec{N})\big).
\end{align}
\end{proof}
The second statement in the limit suggests that the $\Kasym$ is positive definite unless the subspaces spanned by $\mlvec{u}_j$ or $\mlvec{w}_j$ are not full rank, though we do not make use of this fact in the proof of Theorem~\ref{thm:smallest_eig}.
\subsection{Proof of Theorem~\ref{thm:smallest_eig}}
\begin{proof}
For each input state $\mlvec{\rho}_{j}=\mlvec{v}_{j}\mlvec{v}_{j}^{\dagger}$, let $\mlvec{u}_{j}$
and $\mlvec{w}_{j}$ denote the projection of $\mlvec{v}_{j}$ onto the positive
and negative subspaces of the measurement. Since the measurment is updated
throughout the training, $\mlvec{u}_{j}$  and $\mlvec{w}_{j}$ are functions of
time. For a QNN with the scaling factor $\gamma$, the QNN prediction for the
input state $\mlvec{\rho}_{j}$ at time $t$ is
$\hat{y}_{j} = \gamma(\mlvec{u}_{j}^{\dagger}(t)\mlvec{u}_{j}(t) - \mlvec{w}_{j}^{\dagger}(t)\mlvec{w}_{j}(t))$.
Additionally by the normalization of quantum states and the orthogonality of the
training sample, we have
$\mlvec{u}_{j}^{\dagger}(t)\mlvec{u}_{j}(t) + \mlvec{w}_{j}^{\dagger}(t)\mlvec{w}_{j}(t) = \delta_{ij}$,
where $\delta_{ij}$ is the Kronecker delta function.
Combining these two conditions, we can solve that
$\mlvec{u}_{j}^{\dagger}\mlvec{u}_{j} = \frac{1}{2}(1\pm 1/\gamma)$ and
$\mlvec{v}_{j}^{\dagger}\mlvec{v}_{j} = \frac{1}{2}(1\mp 1/\gamma)$ for
$y_{j} = \pm 1$.

By Lemma~\ref{lm:asympkernelestimate}, the diagonal entries
$(\Kasym)_{jj} = 8\gamma^{2}Re(\mlvec{u}_{j}^{\dagger}\mlvec{u}_{j}\mlvec{w}_{j}^{\dagger}\mlvec{w}_{j}) = 8\gamma^{2}\cdot \frac{1}{2}(1\pm 1/\gamma) \cdot\frac{1}{2}(1\mp 1/\gamma) = 2\gamma^{2}(1-1/\gamma^{2})$.

Without loss of generality, assume $y_{1}=y_{2}=\cdots=y_{m/2} = 1$ and
$y_{m/2+1}=y_{m/2+2}=\cdots y_{m} = -1$. Then
$\mlvec{u}_{j}=\sqrt{\frac{1+1/\gamma}{2}}\hat{\mlvec{u}}_{j}$ for
$1\leq j\leq m/2$ and
$\mlvec{u}_{j}=\sqrt{\frac{1-1/\gamma}{2}}\hat{\mlvec{u}}_{j}$ for
$m/2+1\leq j \leq m$. Here $\hat{\mlvec{u}}_{j}$ are unit vectors defined as
$\mlvec{u}_{j}/\sqrt{\mlvec{u}_{j}^{\dagger}\mlvec{u}_{j}}$. For the
off-diagonal entries,
$(\Kasym)_{ij} = 8\gamma^{2}Re(\mlvec{u}_{i}^{\dagger}\mlvec{u}_{j}\mlvec{w}_{j}^{\dagger}\mlvec{w}_{i}) = 8\gamma^{2}Re(\mlvec{u}_{i}^{\dagger}\mlvec{u}_{j} \cdot  (-\mlvec{u}_{j}^{\dagger}\mlvec{u}_{i})) =  -8\gamma^{2}|\mlvec{u}_{i}^{\dagger}\mlvec{u}_{j}|^{2}$.
For the first equality we use the orthogonality among $\{\mlvec{v}_{j}\}_{j=1}^{m}$.

Define $m\times m$ Hermitian $\mlvec{G}$ such that
$G_{ij}=\hat{\mlvec{u}}_{i}^{\dagger}\hat{\mlvec{u}}_{j}$ and $\mlvec{R}$ such
that
$R_{ij} = \frac{1}{2}(1+1/\gamma)$ for $1\leq i,j\leq m/2$,
$R_{ij} = \frac{1}{2}(1-1/\gamma)$ for $m/2+1\leq i,j\leq m$, and
$R_{ij} = \frac{1}{2}\sqrt{1-1/\gamma^{2}}$ for
$1\leq i\leq m/2, m/2+1\leq j\leq m$ or
$m/2+1\leq i\leq m, 1\leq j\leq m/2$. The off-diagonal entries can be expressed $-8\gamma^{2}R_{ij}G_{ij}G_{ji}$.

Using the notations of $\mlvec{R}$ and $\mlvec{G}$, the matrix $\Kasym$ at the
global minima can be expressed as
\begin{align}
\Kasym = 2\gamma^{2}(1-1/\gamma^{2})\mlvec{I} - 8\gamma^{2}\mlvec{R}\circ(\mlvec{G}-\mlvec{I})\circ(\mlvec{G}^{T}-\mlvec{I}),
\end{align}
where $\mlvec{I}$ is the $m\times m$ identity matrix.

\paragraph{Eigenvalues of $\mathbf{R}$}
Let $\mlvec{e}_{1}$ and $\mlvec{e}_{2}$ denote the unit vectors
\begin{align}
  \mlvec{e}_{1} &= \sqrt{\frac{2}{m}}(1,1,\cdots, 1, 0, 0, \cdots, 0)^{T}\\
  \mlvec{e}_{2} &= \sqrt{\frac{2}{m}}(0,0 \cdots, 0, 1,1,\cdots, 1)^{T}
\end{align}
that are zero in the first (last) $m/2$ entries. The matrix $\mlvec{R}$ can be
written as
\begin{align}
  \frac{m}{2}(
  \frac{1}{2}(1+1/\gamma)\mlvec{e}_{1}\mlvec{e}_{1}^{\dagger} +
  \frac{1}{2}(1-1/\gamma)\mlvec{e}_{2}\mlvec{e}_{2}^{\dagger} +
  \frac{1}{2}\sqrt{1-1/\gamma^{2}}\mlvec{e}_{1}\mlvec{e}_{2}^{\dagger}+
  \frac{1}{2}\sqrt{1-1/\gamma^{2}}\mlvec{e}_{2}\mlvec{e}_{1}^{\dagger}
  )
\end{align}
and can be shown to have eigenvalues $(\frac{m}{2},0,\cdots,0)$ by
straight-forward calculation.
\paragraph{Eigenvalues of $\mathbf{G}$}
Over the uniform measure over all the global minima, the direction vectors
$\hat{\mlvec{u}}_{i}$ are sampled independently and uniformly from a
$d/2$-dimensional (complex) sphere. By the approximate isometric properties (see
e.g. Theorem 5.58 in \cite{vershynin2010introduction}), the gram matrix
$\mlvec{G}$ of $\{\hat{\mlvec{u}}_{j}\}_{j=1}^{m}$ is approximately an isometry:
with probability $\geq 1-2\exp(-c_{p}t^{2})$
\begin{align}
  \opnorm{\mlvec{G}-\mlvec{I}} \leq c_{m}\frac{\max\{\sqrt{m}, t\}}{\sqrt{d}}
\end{align}
for constants $c_{p}$ and $c_{m}$.

Applying Lemma~\ref{lm:hadamard} to $\mlvec{R}$, $\mlvec{G}-\mlvec{I}$ and
$\mlvec{G}^{T}-\mlvec{I}$, we have that with probability $\geq 1-\delta$, the
smallest eigenvalues of $\Kasym$ at global minima is greater than or equal to
\begin{align}
 2\gamma^{2}(1-1/\gamma^{2} - C_{2}\max\{\frac{m^{2}}{d}, \frac{m\log(2/\delta)}{d}\})
\end{align}
for some constant $C_2>0$.
\end{proof}

\section{Experiments}
\label{sec:app_exp}
\subsection{Experiment details}
\label{subsec:app_exp_details}
Our numerical experiments involve simulating both quantum neural networks and
the asymptotic dynamics.
\paragraph{QNN simulation}
We simulate the QNN experiments using
Pytorch~\citep{pytorchcite} with the periodic ansatze defined in
Definition~\ref{def:partial-ansatz}. The generating Hamiltonian $\mlvec{H}$ are
chosen to be a $d$-dimensional diagonal matrix with $d/2$ $\sqrt{d-d^{-1}}$ and $d/2$ $-\sqrt{d-d^{-1}}$ on
the diagonal (normalized such that $\tr(\mlvec{H^{2}})/(d^{2}-1)=1$). Each instance of the experiments is specified by the number of
samples $m$, system dimension $d$, number of parameters $p$ and the scaling
factor $\gamma$. A $m$-sample dataset is generated by randomly
sampled $m$ orthogonal pure states $\{\mlvec{v}_{i}\}_{i=1}^{m}\in\complex^{d}$ and randomly assigned
half of the samples with label $+1$ and the other half label $-1$ (i.e. $\{y_{i}\}_{i=1}^{m}\subset\{\pm 1\}^{m}$).

The optimizer we use is the standard gradient descent optimizer. To simulate the
dynamics of gradient flow, we choose the learning rate to be $0.001 / p$ and the
maximum number of epochs is set to be $10000$. We run the experiments on Amazon EC2 C5 Instances.
\paragraph{Asymptotic dynamics simulation}
Theorem~\ref{thm:qnn-mse-linear} allows us to examine the behavior of
QNN dynamics when $p\rightarrow \infty$ by studying the asymptotic dynamics:
\begin{align}
  \frac{d\mlvec{M}(t)}{dt} = - \eta\sum_{j=1}^{m}r_{j}[\mlvec{M}(t), [\mlvec{M}(t), \mlvec{\rho}_{j}]],
  \quad
  \text{where }
  \forall j\in[m], r_{j} := \tr(\mlvec{M}(t)\mlvec{\rho}_{j}) - y_{j}.
\end{align}
For a QNN asymptotic dynamics with number of samples $m$, system dimension $d$
and scaling
factor $\gamma$, we initialize $\mlvec{M}(0)$ as
\begin{align}
\gamma\mlvec{U}
\begin{bmatrix}
  +1 &  0 & \cdots & 0 & 0\\
   0 & +1 & \cdots & 0 & 0\\
   \vdots & \vdots & \ddots & \vdots & \vdots\\
   0 & 0 & \cdots & -1 & 0\\
   0 & 0 & \cdots & 0& -1
\end{bmatrix}
\mlvec{U}^{\dagger}
\end{align}
with $\mlvec{U}$ being a $d\times d$ haar random unitary. Similar to the QNN
simulation, the training set is chosen to be $m$ orthogonal pure states with labels randomly sampled from $\{\pm 1\}$.
The simulation of the asymptotic dynamics is run on Intel Core i7-7700HQ Processor (2.80Ghz) with 16G memory.
\subsection{\texorpdfstring{$\Kasym$}{Kasym} as a function of \texorpdfstring{$t$}{t}}
In Corollary~\ref{cor:onesample}, we see that the convergence rates for
one-sample QNNs change significantly during training.
Theorem~\ref{thm:qnn-mse-linear} allows us further verify this observation for
training sets with $m>1$ by simulating the asymptotic dynamics.

In Figure~\ref{fig:reKchg_vary}, we plot the relative change of the
$\mlvec{K}_{asym}(t)$ defined as
\begin{align}
  (\mlvec{K}_{asym}(t))_{ij}:=\tr\big(i[\mlvec{M}(t), \mlvec{\rho}_{i}] i [\mlvec{M}(t), \mlvec{\rho}_{j}]\big).
\end{align}
Each of the data point is averaged over 100 random initialization of
$\mlvec{M}(0)$. It is observed that $\mlvec{K}_{asym}(t)$ changes significantly
($\geq 5\%$) for each of the hyperparameters $d$, $m$ and $\gamma$. Therefore we
conclude that the deviation from the neural tangent kernel regression is
ubiquitous in general for practical settings. Particularly it rules out the
existing belief that the $d\rightarrow\infty$ alone can lead to a neural tangent
kernel-like behavior in QNNs. Same is observed for over-parameterized QNNs (Figure~\ref{fig:Kchg_vary_qnn})

\begin{figure}[!htbp]
  \centering
  \includegraphics[width=.9\linewidth]{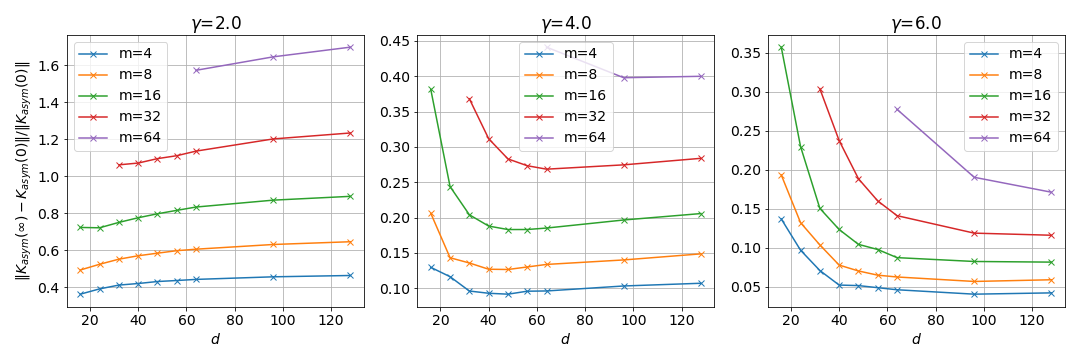}\\
  \includegraphics[width=.91\linewidth]{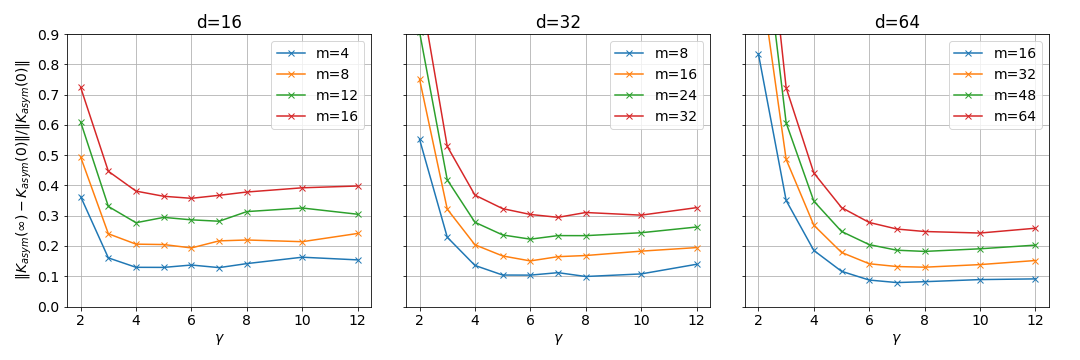}
  \includegraphics[width=.91\linewidth]{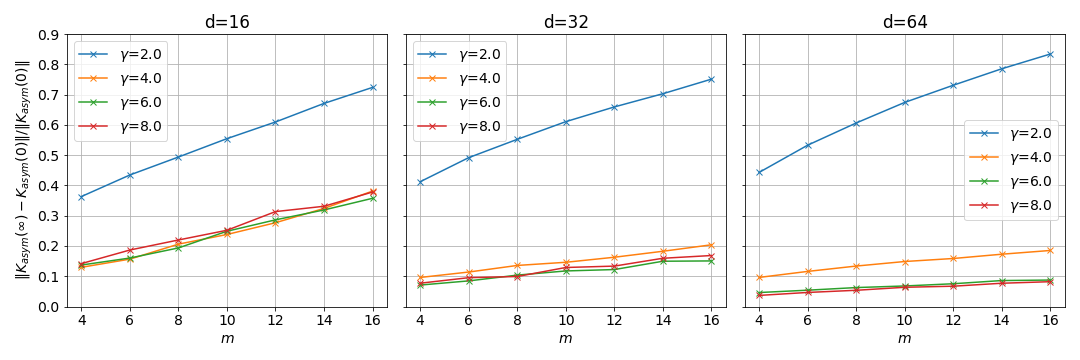}
  \caption{
    Relative change of $\Kasym(t)$ in the QNN asymptotic dynamics
    for varying system dimension $d$, scaling factor $\gamma$ and number of
    training samples $m$.
    $\mlvec{K}_{asym}(t)$ changes significantly ($\geq 5\%$) throughout training.
  }
  \label{fig:reKchg_vary}
\end{figure}

\begin{figure}[!htbp]
  \centering
  \includegraphics[width=.91\linewidth]{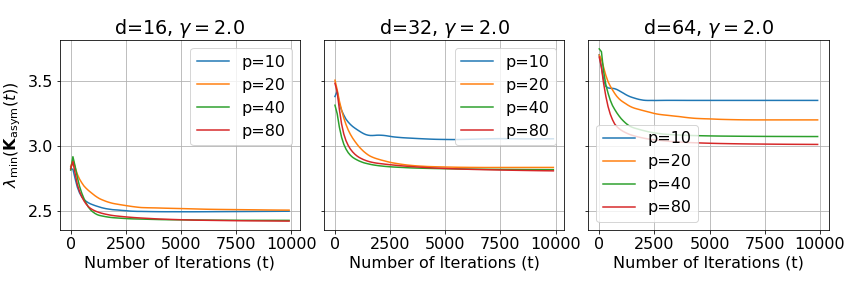}
  \caption{
     Change of the $\lambda_{\min}(\Kasym(t))$ during the training in QNNs
     with $m=4, \gamma=2.0$ and varying $d$.
  }
  \label{fig:Kchg_vary_qnn}
\end{figure}


\end{document}